\documentclass[onefignum,onetabnum]{siamart190516}
\usepackage{amssymb,amsmath}
\usepackage{float}
\usepackage{algorithm}
\usepackage{algorithmic} 
\usepackage{graphics} 
\usepackage{caption}
\usepackage{subfigure}
\usepackage{multirow}
\usepackage{url}
\usepackage{mathdots}
\usepackage{arydshln}
\usepackage{bm}
\usepackage{ntheorem}
\usepackage{appendix}
\newsiamthm{assumption}{assumption}
\newcommand{\tabincell}[2]{\begin{tabular}{@{}#1@{}}#2\end{tabular}}

\usepackage{lipsum}
\usepackage{amsfonts}
\usepackage{graphicx}
\usepackage{epstopdf}
\usepackage{algorithmic}
\ifpdf
  \DeclareGraphicsExtensions{.eps,.pdf,.png,.jpg}
\else
  \DeclareGraphicsExtensions{.eps}
\fi

\newsiamremark{remark}{Remark}
\newsiamremark{hypothesis}{Hypothesis}
\crefname{hypothesis}{Hypothesis}{Hypotheses}
\newsiamthm{claim}{Claim}

\headers{Fourier Phase Retrieval with background information}{Ziyang Yuan, Hongxia Wang, Haoxing Yang, Ningyi Leng}

\title{An Example Article\thanks{Submitted to the editors DATE.
\funding{This work was funded by the Fog Research Institute under contract no.~FRI-454.}}}

\author{Dianne Doe\thanks{Imagination Corp., Chicago, IL 
  (\email{ddoe@imag.com}, \url{http://www.imag.com/\string~ddoe/}).}
\and Paul T. Frank\thanks{Department of Applied Mathematics, Fictional University, Boise, ID 
  (\email{ptfrank@fictional.edu}, \email{jesmith@fictional.edu}).}
\and Jane E. Smith\footnotemark[3]}

\usepackage{amsopn}

\ifpdf
\hypersetup{
  pdftitle={Phase Retrieval with Background Information: Decreased References and efficient methods},
  pdfauthor={Ziyang Yuan and Hongxia Wang,e.t.al}
}
\fi

\begin{document}
	\bibliographystyle{unsrt}
    	\title{Phase Retrieval with Background Information: Decreased References and Efficient methods}
	
	\author{Ziyang Yuan\thanks{First author,
			Academy of Military Sciences, Beijing and Department of Mathematics, National University of Defense Technology,
			Changsha, Hunan, 410073, P.R.China.(\email{yuanziyang11@nudt.edu.cn})}
		\and HaoXing Yang\thanks{Second author,
				Department of Mathematics, National University of Defense Technology,
				Changsha, Hunan, 410073, P.R.China.  (\email{saber1222@foxmail.com})}
		\and NingYi Leng\thanks{Third author,
				Department of Mathematics, National University of Defense Technology,
				Changsha, Hunan, 410073, P.R.China.  (\email{623567098@qq.com})}
    \and Hongxia Wang\thanks{Corresponding author,
				Department of Mathematics, National University of Defense Technology,
				Changsha, Hunan, 410073, P.R.China.  Corresponding author.(\email{wanghongxia@nudt.edu.cn})}
	}
\maketitle
\begin{abstract}
 	Fourier phase retrieval(PR) is a severely ill-posed inverse problem that arises in various applications. To guarantee a unique solution and relieve the dependence on the initialization, background information can be exploited as a structural priors \cite{Yuan_2019}. However, the requirement for the background information may be challenging when moving to the high-resolution imaging. At the same time, the previously proposed projected gradient descent(PGD) method also demands much background information. 
  
  In this paper, we present an improved theoretical result about the demand for the background information, along with two Douglas Rachford(DR) based methods. Analytically, we demonstrate that the background required to ensure a unique solution can be decreased by nearly $1/2$ for the 2-D signals compared to the 1-D signals. By generalizing the results into $d$-dimension, we show that the length of the background information more than $(2^{\frac{d+1}{d}}-1)$ folds of the signal is sufficient to ensure the uniqueness. At the same time, we also analyze the stability and robustness of the model when measurements and background information are corrupted by the noise. 
  Furthermore, two methods called Background Douglas-Rachford (BDR) and Convex Background Douglas-Rachford (CBDR) are proposed. BDR which is a kind of non-convex method is proven to have the local R-linear convergence rate under mild assumptions. Instead, CBDR method uses the techniques of convexification and can be proven to own a global convergence guarantee as long as the background information is sufficient. To support this, a new property called F-RIP is established. We test the performance of the proposed methods through simulations as well as real experimental measurements, and demonstrate that they achieve a higher recovery rate with less background information compared to the PGD method.
\end{abstract}


\section{Introduction}

	The Fourier PR problem, which aims to recover a signal from its intensity only Fourier spectrum, arises in a variety of applications such as X-ray crystallography, astronomy, coherent diffraction imaging(CDI) and Fourier ptychography\cite{Miao1999Extending,fienup1987phase,Zheng2013Wide}. 
	
\subsection{Problem formulation}

Mathematically, the 1-D discrete Fourier PR can be formulated as 
\begin{eqnarray}\label{model}
&\mathrm{Find}~\mathbf{x\in\mathbb{C}^n}\nonumber\\
&\mathrm{s.t.}~|\mathbf{F}_i^{\textrm{H}}\mathbf{x}|^2=b_i,~i=1,\cdots,m,
\end{eqnarray}
where $\{\mathbf{F}_i:=[1,e^{\frac{2\pi j(i-1)}{m}},\cdots,e^{\frac{2\pi j(i-1)(n-1)}{m}}]^\text{T}, i=1,\cdots,m\}$ is the Fourier basis with $j=\sqrt{-1}$ , $[\cdot]^{\text{H}}$ denotes the conjugate transpose. $\mathbf{b} = [b_1,b_2,\cdots,b_m]^\text{T}$ is the measurement.  	In short, the Fourier PR problem is to recover the signal of interest $\mathbf{x}$ from the intensity only Fourier spectrum $\mathbf{b}$.

\indent
It is evident that the solutions of \eqref{model} are not unique. For example, if $\mathbf{x}_1=\mathbf{x}_2e^{j\mathbf{\theta}}$, where $\theta\in[0,2\pi)$, then $|\mathbf{F}_i^{\mathrm{H}}\mathbf{x}_1|=|\mathbf{F}_i^{\mathrm{H}}\mathbf{x}_2|$, $i=1,\cdots,m.$
Additionally, the conjugate transpose and circular shift of $\mathbf{x}$ also yield the same measurements $\mathbf{b}$. Therefore, the uniqueness of the Fourier PR problem is considered up to these transformations which are called trivial solutions. Recent studies have focused on understanding the properties of the solutions and how to mitigate the effects caused by the non-trivialities\cite{Beinert2016Enforcing,grohs2020phase,bendory2017fourier,beinert2015ambiguities,Hofstetter1964Construction}. 

\subsection{Previous works}

\indent
 For the 1-D Fourier PR problem \eqref{model}, it has been proven that there are at most $2^{n-1}$ different solutions up to a global phase when $m\geq2n-1$\cite{Hofstetter1964Construction}. However, the results change for the problem in the multiple dimension. According to \cite{Hayes1982The}, if the dimension $d\geq2$ and $m_i\geq 2n_i-1,i=1,\cdots,d$, the uniqueness can be guaranteed except for a set of signals with zero measure, where $n_i$ and $m_i$ are the length of the signal and measurements respectively in each dimension. This sheds light on the applications of the 2-D Fourier PR problem. A series of algorithms has been subsequently been developed. Representatives are the alternating projection or reflection methods such as GS method, HIO method, RAAR method and ADMM method.\cite{Fienup1982Phase,Gerchberg1971A,wen2012alternating,luke2004relaxed,bauschke2003hybrid}. These methods update iterations by projection or reflection between different domains. However, these methods are liable to stagnate namely iterations generated by the algorithms stagnate before converging to the neighbor of the ground truth  especially when the measurements are corrupted by the noise. At the same time, although many theoretical and algorithmic results about PR for the general frames\footnote{General frames here mean the 
transforming vectors $\mathbf{F}_i,i=1,\cdots,m$ in \eqref{model1} are not orthogonal basis but a kind of frames.} have made a great progress\cite{grohs2020phase,Balan2006On,Candes2014Phase,grohs2019stable}, these works usually fail when dealing with the Fourier PR problem. They are instructive but cannot be applied to the Fourier PR problem directly. 

 Utilizing priors is one of the efficient ways to constrain the feasible solutions and relieve the illness of the Fourier PR problem. For instance, additional information about the unknown signal may be helpful to guarantee the unique solution \cite{kim1990phase,langemann2008phase,seifert2006multilevel}. Sparsity is one of the widely used priors. It has been shown in \cite{Ranieri2013Phase} that the full auto-correlation series of a 1-D $s$-sparse  real signal \footnote{$s$-sparse signal means there are up to $s$ nonzero elements in a signal.} $\mathbf{x}$ is sufficient to uniquely determine $\mathbf{x}$ as long as $s\neq 6$ and the auto-correlation sequence is collision free
\footnote{A signal $\mathbf{x}\in\mathbb{R}^n$ is said to be collision free if for any distinct $i_1$, $i_2$, $i_3$, $i_4\in$ $1,2,\cdots,n$, we have $x_{i_1}-x_{i_2}\neq x_{i_3}-x_{i_4}$.}. 
 In addition, if $m$ is a prime, \cite{Ohlsson2013On} proves that $m\geq s^{2}-s+1$ measurements are sufficient to uniquely determine the auto-correlation series of $\mathbf{x}$. For the $d$-dimension signal $\mathbb{R}^{n\times\cdots\times n}$, under the condition that the support of the signal is distributed as the Gaussian process, 
\cite{novikov2021support} proves that the signal can be recovered from the auto-correlation function of the signal with high probability as long as the sparsity is at most $\mathcal{O}(n^{d\theta})$ for any $\theta<\frac{1}{2}.$
Meanwhile, a large number of highly efficient optimization algorithms are come up to deal with the sparse Fourier PR problem \cite{Szameit2012Sparsity,Jaganathan2012Recovery,Shechtman2013GESPAR}. Interested readers can refer to \cite{Beinert2016Enforcing, beinert2015ambiguities} for a comprehensive reviews about the different priors which may enforce the Fourier PR problem to be unique. In \cite{kim1990phase,2013Vectorial}, they attempt to use interference with a known or unknown reference signal to achieve the uniqueness
of the Fourier PR problem. In \cite{2019Holographic}, they propose a general mathematical framework and recovery algorithm for the PR problem under different references.  Notice that in \cite{Learned}, it uses the deep neural network to learn a data-adaptive reference so that algorithm proposed can perform better.
 
While a series of priors have been introduced, most of them only guarantee uniqueness up to trivial solutions. This means that it is impossible to distinguish $\bm{x}$ from its trivial ambiguities, which significantly degrades the quality of the reconstruction. Therefore, more efficient priors are required to eliminate this issue. In \cite{Yuan_2019}, a model that utilizes background information is proposed as follows: 
\begin{eqnarray}\label{model1}
&\mathrm{Find}~\mathbf{z}\in\mathbb{C}^{n+k}\nonumber\\
&\mathrm{s.t.}~|\mathbf{F}_i^{\textrm{H}}\mathbf{z}|^2=b_i,~i=1,\cdots,m,\nonumber\\
&z_{n+l}=y_l,~l=1,\cdots,k,&
\end{eqnarray}
where $\mathbf{y}\in\mathbb{R}^k$ represents the background information that is known in advance.
Utilizing this information, the uniqueness of Fourier PR can be  guaranteed  if $k\geq3n-1$, $y_l\overset{i.i.d}{\sim}\mathcal{N}(0,\sigma^2)$, where $\sigma^2>0$, and $l=1,2,\cdots,k$. It should be noted that the solution here is \textbf{exactly unique}. Numerical tests have shown that the projected gradient descent(PGD) method can recover the ground truth without requiring an elaborate initialization and is robust to the noisy measurements. However, the demand for the background information about $3n$ is still large in the practical applications, especially when the dimension is high.
Recall that for the Fourier PR problem with dimension $d$, $m_i\geq2n_i-1,i=1,2\cdots,d$ is required to ensure uniqueness up to the trivial solutions\cite{Hofstetter1964Construction}. However, according to \cite{miao1998phase}, samples less than $2n_i-1,i=1,2\cdots,d$ is also sufficient, and $m_i\geq2^{\frac{1}{d}}n_i,i=1,2,\cdots,d$ can work in practice. This conjecture hints us to utilize the correlations between different dimensions to lower the demand of background information.

\subsection{Works in this paper}
This paper presents theories about the uniqueness, stability and robustness of \eqref{model1}. These results require less background information compared to \cite{Yuan_2019}. Additionally, we also proposes two new methods that guarantee local and global convergence respectively. Meanwhile, numerical simulations and optical experiments demonstrate the efficiency of the proposed methods. Specifically, works in this paper can be classified into four main aspects.
\begin{itemize}
	\item Firstly, we prove the uniqueness of the solution for \eqref{model1} in the 2-dimension situation namely ground truth $\mathbf{X}\in\mathbb{R}^{n_1\times n_2}$. Specifically, we show that $ (n_1+k_1)(n_2+k_2)\geq(2n_1-1)(3n_2-1)+n_2$ ensures the uniqueness when each element of the background $\mathbf{Y}$ satisfies the normal distribution, where $k_1$ and $k_2$ represent the length of the background in two directions, respectively. Furthermore, we generalize the results into $d$-dimension and prove that $$\prod_{i=1}^d (n_i+k_i)\geq2\prod_{i=1}^d n_i+\prod_{i=1}^d(2n_i-1),i=1,2,\cdots,d,$$ is sufficient to guarantee the uniqueness. These results also develop the theories of affine Fourier phase retrieval, which can seen in Table \ref{TB}. 
	\item Secondly, we prove the stability of  \eqref{model1} in the 2-dimension situation, specifically
	$$C_1\|\mathbf{X}_1-\mathbf{X}_2\|_{\mathrm{F}}\leq\|\mathrm{vec}(\mathbf{I}_1)-\mathrm{vec}(\mathbf{I}_2)\|_2\leq C_2\|\mathbf{X}_1-\mathbf{X}_2\|_{\mathrm{F}},$$
	where $C_1$ and $C_2$ are constants determined by $n_1$, $n_2$ and $\mathbf{Y}$. $\mathbf{I}_1$ and $\mathbf{I}_2$ are the intensity-only measurements generated by $\mathbf{X}_1$ and $\mathbf{X}_2$ with the same $\mathbf{Y}$. Moreover, the error bound of the distance between the estimation $\mathbf{X}^*$ and the ground truth $\mathbf{X}$ is also derived when the background information and intensity-only measurements are corrupted by the noise $\varepsilon_1$ and $\varepsilon_2$ respectively and $\tilde{\mathbf{Y}}$ and $\tilde{\mathbf{I}}$ are known:
$$\|\mathbf{X}^*-\mathbf{X}\|_{\mathrm{F}}\leq C_1(c_1,c_2)+C_2(c_2)\|\text{vec}(\tilde{\mathbf{Y}})\|_{1}+\sqrt{C_3(c_2)\|\text{vec}(\tilde{\mathbf{I}})\|_{1}+C_4(c_1,c_2)},$$	
where $c_1$ and $c_2$ are the noise levels of $\varepsilon_1$ and $\varepsilon_2$ satisfying $\|\bm{\varepsilon}_1\|_{\infty}\leq c_1$ and $\|\bm{\varepsilon}_2\|_{\infty}\leq c_2$, and $C_1(c_1,c_2)$, $C_2(c_2)$, $C_3(c_2)$, $C_4(c_1,c_2)$ are constants determined by $c_1$ and $c_2$.
	\item Next, to improve the performance of the PGD method when the background information is insufficient, a method called the Background Douglas Rachford(BDR) is proposed. Theoretical results demonstrate that it can escape the fixed points of the method, which may not be the ground truth. Furthermore, under mild assumptions, we prove that the BDR method has an R-linear convergence rate if the initialization is close to the ground truth
    \item At the same time, we also use convex relaxation techniques, transforming \eqref{model1} into a convex problem \eqref{conbkmodel}. Theoretical results show that \eqref{conbkmodel} and \eqref{model1} can keep the same solution with probability $1$ as if background information is sufficient. To demonstrate this, a property called Fourier Spectrum Restricted Isometry Property (F-RIP) is established. 
    \begin{eqnarray*}
(c_1(n,k,\mathbf{x})-\delta)\|\mathbf{h}\|_2^2\leq\frac{1}{k}\|\mathbf{L}_1\mathbf{h}\|_2^2\leq(c_2(n,k,\mathbf{x})+\delta)\|\mathbf{h}\|_2^2,\forall\mathbf{h}\in\mathbf{C}_2,
\end{eqnarray*}
where $c_1(n,k,\mathbf{x})$ and $c_1(n,k,\mathbf{x})$ are two constants depending on $n,k$ and ground truth, $c_1(n,k,\mathbf{x})>\delta$, $\mathbf{C}_2=\{\mathbf{h}\in\mathbb{R}^{n+k}| |\mathbf{F}^{\mathrm{H}}_i\mathbf{h}|^2\leq4,i=1,\cdots,n+k\},$ and $\mathbf{L}_1\in\mathbb{R}^{k\times(n+k)}$ is a special partial random circulant matrix where part of the variables in $\mathbf{L}_1$ are determined. 
    Additionally, we introduce the Convex Background Douglas Rachford (CBDR) method for solving the convex problem \eqref{conbkmodel}, which has a guaranteed global convergence.
	\item Last, through numerous tests including simulated Fourier measurements and optical Charge-coupled Device(CCD) measurements, we have verified the high efficiency of the newly proposed methods which require less background information to achieve a higher successfully rate compared to the PGD method.
\end{itemize}
\indent 
In this article, bold uppercase and lowercase letters represent matrices or sets and vectors respectively.  $|\cdot|$ denotes the module of a number, the determinant of matrix or the cardinality of a set. The notation $\textrm{vec}(\cdot)$ represents the vectorizing operation, while $\odot$ denotes the Hadamard product and  $\otimes$ represents the Kronecker product. The abbreviation $\textrm{i.i.d}$ stands for 'independent and identically distributed', and
$\|\cdot\|_{\mathrm{F}}$ is the Frobenius norm. The symbol $\|\cdot\|_1$ denotes the $1$-norm of a vector or matrix. $\mathbf{B}_1^{n+k}$ is the $\ell_1$ Ball in $\mathbb{R}^{n+k}$, and $\mathbf{B}_2^{n+k}$ is the $\ell_2$ Ball in $\mathbb{R}^{n+k}$. 
\section{Theoretical Analysis of the model}
For the sake of brevity, we only analyze the problem when the signal of interest is real, but the complex situation can be generalized. The theoretical analysis is specialized to the multi-dimension array, thus we don't use the same notations as in \cite{Yuan_2019}.
\subsection{Uniqueness of the solution}
The physical setup for the phase retrieval with background information, as studied in this paper, is depicted in Figure \ref{psm}. An unknown sample $\mathbf{X}\in\mathbb{R}^{n_1\times n_2}$ is centered on a reference $\mathbf{Y}\in\mathbb{R}^{(n_1+k_1)\times (n_2+k_2)}$. The region of the reference that overlaps with the sample is transparent, i.e., $\mathbf{Y}|_{\bm{\Omega}}=\mathbf{0}$ where $\bm{\Omega}$ is the support of $\mathbf{X}$ and is known in advance. The reference object can be an SLM (Spatial Light Modulator), which is set to modulate the light passing through the overlapped area. The CCD records only the phase-free intensity  $\mathbf{I}\in\mathbb{R}^{m_1\times m_2}$ which is the diffraction pattern resulting from the  combination of $\mathbf{X}$ and $\mathbf{Y}$.


If we denote the combination of $\mathbf{X}$ and $\mathbf{Y}$ as $\mathbf{Z}\in\mathbb{R}^{(n_1+k_1)\times (n_2+k_2)}$, that is to say,
$$
\mathbf{Z}_{i,h} = \begin{cases}
\mathbf{X}_{i,h}, & (i, h)\in \bm{\Omega} \\
\mathbf{Y}_{i,h}, & \mathrm{otherwise}.
\end{cases}
,$$
then the mathematical model of the PR problem with background information can be formulated as:
\begin{eqnarray}
\begin{array} { c } 
{ \textrm{ Find } \mathbf { Z } } \\ { \text { s.t. }  \left| \mathbf { F}_{(n_1+k_1)\times m_1} ^ {\text{H}} \mathbf { Z } \mathbf { F}_{m_2\times(n_2+k_2)}^{\text{H}} \right| ^ { 2 } }=\mathbf { I }  \\
{\mathbf { Z }_{i , h} = \left\{ \begin{array} { l l } { \mathbf { X }_{i , h}, } & { ( i , h ) \in\bm{\Omega }} \\ { \mathbf { Y }_{i , h} , } & { \text { otherwise } } \end{array} \right.,}
\end{array}\label{e15}
\end{eqnarray}
where 
$\mathbf{F}_{(n_1+k_1)\times m_1}^{\text{H}}\in\mathbb{C}^{m_1\times (n_1+k_1)}$ and $\mathbf{F}_{m_2\times(n_2+k_2)}^{\text{H}}\in\mathbb{C}^{ (n_2+k_2)\times m_2}$ are sub-matrices extracted  separately from Fourier matrices $\mathbf{F}_{m_1\times m_1}^{\text{H}}$ and $\mathbf{F}_{m_2\times m_2}^{\text{H}}$.

\begin{figure}
	\includegraphics[width = 1\textwidth]{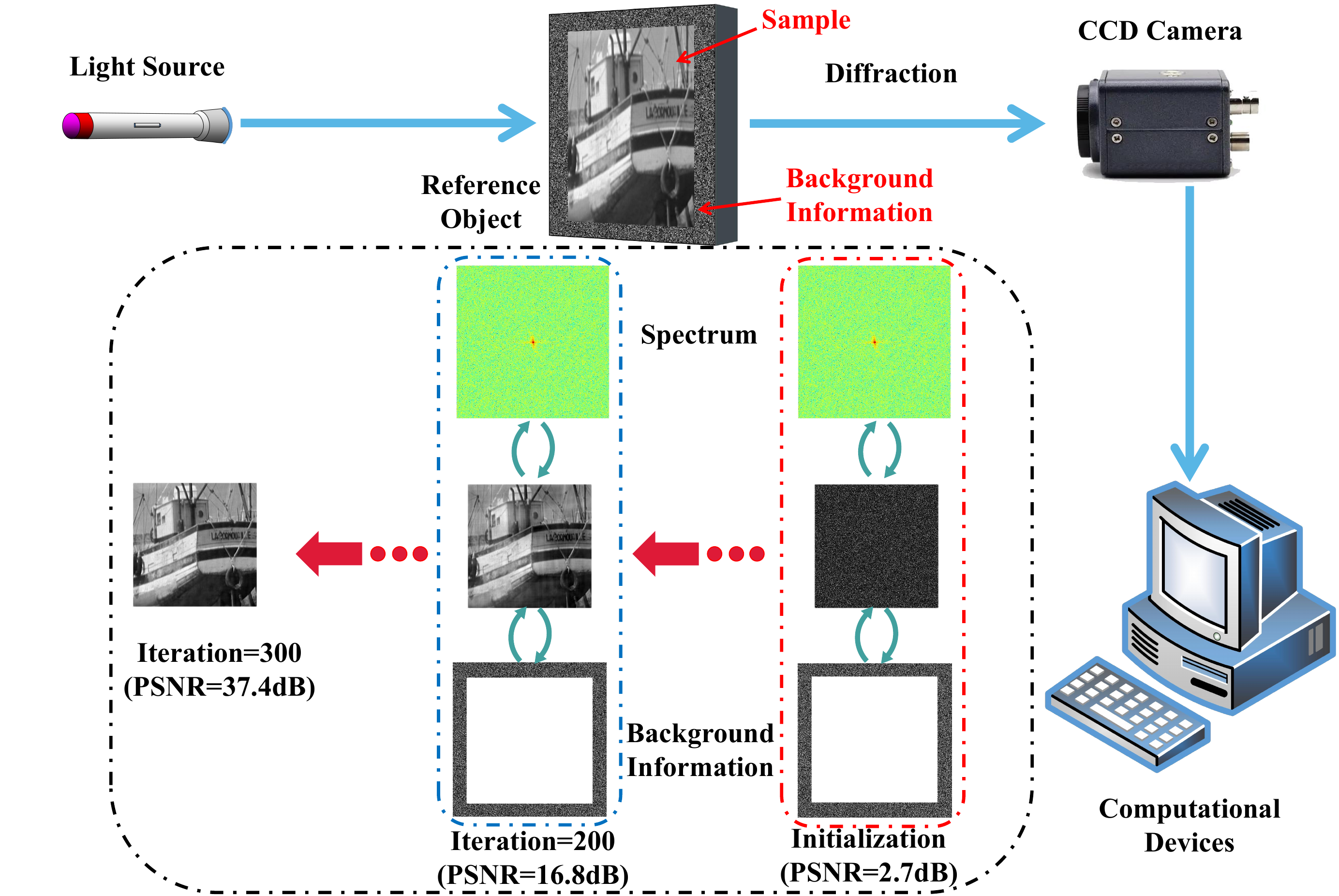}
	\centering
	\caption{Procedure of phase retrieval with background information, where PSNR stands for the Peak Signal to Noise Ratio.}
	\label{psm}
\end{figure}

The generalization of the theorem in \cite{Yuan_2019} states that if one of the following two conditions is satisfied, then (\ref{e15}) can guarantee the exact uniqueness of the solution namely excluding all trivial ambiguities.\\[5pt]
\begin{itemize}
	\item[1.]  $m_1\geq2(n_1+k_1)-2$, $k_1\geq n_1$, $k_2\geq0$ and $\mathbf{Z}_{n_1+k_1, h}\neq0,~h=0,\cdots,n_2-1$.
	\item[2.] $m_1=n_1+k_1+p$, $p\geq 0$, $k_2\geq0$,
	\begin{eqnarray*}
		k_1\geq
		\left\{
		\begin{aligned}
			&\text{max}~(3n_1-2-p,n_1),~m_1-2n_1+1\text{~is odd}&\\
			&\text{max}~(3n_1-1-p,n_1),~m_1-2n_1+1\text{~is even}&
		\end{aligned}
		\right.
	\end{eqnarray*}
	and $\mathbf{Y}_{i,h}\overset{i.i.d}{\sim}\mathcal{N}(\mu,\sigma^2)$, $\sigma^2>0$, $(i,h)\notin\bm{\Omega}$.\\[5pt]
\end{itemize}

When there is no oversampling, i.e., $m_1=n_1+k_1$, the result above shows that $k_1$ should be at least $3$ times as large as $n_1$. This requirement may exceed the resolution of CCD in practice especially when $n_1$ is large. Note that there are no extra requirements for $k_2$. Intuitively, we can constrain $k_2$ by utilizing the correlations of $\mathbf{X}$ along two directions, which can reduce the requirement for the background information. As stated in \cite{Yuan_2019}, oversampling can generally decrease the requirements for the background information. Thus, Theorem \ref{t9} only consider the worst-case scenario when $m_i=n_i+k_i,i=1,2$.
\begin{theorem}\label{t9}
	Consider model \eqref{e15} with $\mathbf{Y}_{i,h}\overset{i.i.d}{\sim}\mathcal{N}(\mu,\sigma^2)$ and $\sigma>0$. If $(n_1+k_1)(n_2+k_2)\geq(2n_1-1)(3n_2-1)+n_2$, the unique solution of (\ref{e15}) can be guaranteed. Especially, if $n=n_1=n_2$ and $k=k_1=k_2$, then $k\geq(\sqrt{6}-1)n\approx1.45n$ is sufficient.
\end{theorem} 

\remark{Comparing to the traditional Fourier PR problem, although the number of measurements required to guarantee the uniqueness of \eqref{e15} is not reduced, namely, $m_i\geq\sqrt{6}n_i\geq2n_i-1,i=1,2$, the solution can be free from trivial ambiguities. This usually makes methods easier to find the ground truth. Additionally, it is worth noting that \eqref{e15} is a special case of the affine phase retrieval problem, achieved by splitting $\mathbf{Z}$ into $\mathbf{X}$ and $\mathbf{Y}$. The 1-D affine phase retrieval problem is formulated as follows:} 
\begin{eqnarray}
\begin{array} { c } 
{ \textrm{ Find } \mathbf {x} } \\ { \text { s.t. } \mathbf {b} = \left| \mathbf { A}^{\text{H}} \mathbf {x}+\mathbf{d} \right| ^ { 2 }}, 
\end{array}\label{ea15}
\end{eqnarray}
where $\mathbf{A}\in\mathbb{C}^{n\times m}$ is the measurement matrix, $\mathbf{b}\in\mathbb{R}^m$ and $\mathbf{d}\in\mathbb{C}^m$ are known in advance.
There are several studies that investigate the uniqueness of \eqref{ea15} under different conditions. Some of these results are summarized in Table \ref{TB}. It can be observed that the affine phase retrieval problem is guaranteed to be unique if $m\geq 4n$, regardless of whether the measurements are general or Fourier.

\begin{table}
\caption{Summary of uniqueness results. It should be noted that the uniqueness of the phase retrieval problem refers to uniqueness up to trivial solutions. However, for affine phase retrieval, uniqueness means that the solution is exactly unique.}
\label{TB}
\centering
\begin{tabular}{|l|l|l|l|}
\hline  & \multicolumn{1}{c|}{\tabincell{c}{Type of\\the signal}} & \multicolumn{1}{c|}{Phase retrieval}  & Affine phase retrieval  \\ \hline
\multicolumn{1}{|c|}{\multirow{3}{*}{\begin{tabular}[c]{@{}c@{}}\\\\\\~\\General \\ measurement\end{tabular}}}
 &\multicolumn{1}{c|}{$\mathbf{x}\in\mathbb{R}^n$}&\tabincell{c}{$m\geq2n-1$, a full-spark\\ random $\mathbf{A}\in\mathbb{R}^{m\times n}$\\ guarantees uniqueness\\ with high probability\cite{Balan2006On}} & \multicolumn{1}{c|}{\tabincell{c}{$m\geq2n+1$ , a generic\\ $(\mathbf{A}, \mathbf{d})\in\mathbb{R}^{m \times(n+1)}$ is\\ sufficient to guarantee\\ the uniqueness\cite{gao2018phase}}}  \\ \cline{2-4} 
&\multicolumn{1}{c|}{$\mathbf{x}\in\mathbb{C}^n$}& \multicolumn{1}{c|}{\tabincell{c}{$m\geq4n-4$, a generic\\$\mathbf{A}\in\mathbb{C}^{m\times n}$ is\\ sufficient to guarantee\\ the uniqueness\cite{Bandeira2013Saving,conca2015algebraic}}}  & \multicolumn{1}{c|}{\tabincell{c}{$m\geq4n+1$, a generic\\$(\mathbf{A}, \mathbf{d})\in\mathbb{C}^{m\times(n+1)}$ is\\ sufficient to guarantee\\ the uniqueness\cite{gao2018phase}}}  \\ \cline{2-4} 
& \tabincell{c}{$s-$sparse\\ $\mathbf{x}\in\mathbb{R}^n(\mathbb{C}^n)$}&\multicolumn{1}{c|}{\tabincell{c}{$m\geq4s-1(\text{or}~8s-2)$,\\ a generic\\ $\mathbf{A}\in\mathbb{R}^{m\times n}(\text{or}~\mathbb{C}^{m\times n})$\\ is sufficient to guaran-\\tee the uniqueness\cite{2012Sparse}}}&\multicolumn{1}{c|}{\tabincell{c}{$m\geq2s+1(\text{or}~4s+1)$,\\ a generic $(\mathbf{A}, \mathbf{d})\in$\\$ \mathbb{R}^{m \times(n+1)}(\text{or}~\mathbb{C}^{m \times(n+1)})$\\
is sufficient to guaran-\\tee the uniqueness\cite{gao2018phase}}} 
\\ \hline
\multirow{2}{*}{\tabincell{c}{\\\\Fourier\\ measurement}} &\multicolumn{1}{c|}{$\mathbf{x}\in\mathbb{C}^n$}   & \multicolumn{1}{c|}{\tabincell{c}{No uniqueness\cite{Hofstetter1964Construction}}}& \multicolumn{1}{c|}{\tabincell{c}{$m\geq4n-1$, a generic\\ $\mathbf{d}\in\mathbb{R}^m$ is sufficient to \\ guarantee unique-\\ness\cite{Yuan_2019}}}\\ \cline{2-4} 
&\multicolumn{1}{c|}{\tabincell{c}{$\mathbf{x}\in$\\$
\mathbb{C}^{n_1\times\cdots\times n_d}$\\$d\geq2$}}&\multicolumn{1}{c|}{\tabincell{c}{$m_i\geq2n_i-1,i=1,2$\\ $,\cdots,d$ is sufficient to \\guarantee the unique-\\ness for real non-redu-\\cible signals\cite{Hayes1982The}}} &\multicolumn{1}{c|}{\tabincell{c}{ $\prod_{i=1}^d m_i\geq2\prod_{i=1}^d n_i+$\\ $\prod_{i=1}^d(2n_i-1)$,a gene-\\ric $\mathbf{d}\in\mathbb{C}^{m_1\times m_2\times\cdots m_d}$,\\ is sufficient to \\ guarantee the uni-\\queness\textbf{[This paper]}}}  \\ \hline
\end{tabular}
\end{table}

Before proving Theorem \ref{t9}, some preliminaries will be introduced. Let $m_i=n_i+k_i$, $i=1,2.$ First, the 2-D circular auto-correlation series of $\mathbf{X}\in\mathbb{R}^{n_1\times n_2}$ is defined as:
\begin{eqnarray}\label{124}
\mathbf{R}_{l_1,l_2} := \mathrm{vec}(\mathbf{X}(-l_1,-l_2))^{\text{T}}\mathrm{vec}(\mathbf{X}),\tabincell{c}0\leq l_1\leq m_1-1,0\leq l_2\leq m_2-1,
\label{e14}
\end{eqnarray}
where $\mathrm{vec}(\mathbf{X})$ is the vectorization of $\mathbf{X}$ obtained by stacking each row of $\mathbf{X}$, and $\mathbf{R}_{l_1,l_2}$ is the element in the $l_1$th row and $l_2$th column of $\mathbf{R}$. Furthermore, $\mathbf{X}_{i,h}(-l_1,-l_2)=\mathbf{X}_{i-l_1,h-l_2}=\mathbf{X}_{i_1,h_1}$,  $i_1\in\{0,\cdots,n_1-1\},h_1\in\{0,\cdots,n_2-1\}$, $i-l_1\equiv i_1\bmod m_1$, $h-l_2\equiv h_1\bmod m_2$. The Wiener-Khinchin Theorem leads to the following lemma, which establishes a connection between $\mathbf{R}$ and $\mathbf{I}$.
\begin{lemma}\label{theoremequation1}
	For the problem shown in \eqref{e15}, the spectrum $\mathbf{I}$ and circular auto-correlation function $\mathbf{R}$ are related by the equation below:
	\begin{equation}\label{theoremequation}
	\mathbf{F}_{m_2\times m_2}^{\mathrm{H}}\mathbf{R}\mathbf{F}_{m_1\times m_1}^{\mathrm{H}}=\mathbf{I}.
	\end{equation}
\end{lemma}

Now, we will prove Theorem \ref{t9}. The main idea is to find sufficient linear combinations of the elements of $\mathbf{X}$ from the circular auto-correlation matrix $\mathbf{R}$.
\begin{proof}[Proof of Theorem 2.1]
By using Lemma \ref{theoremequation1}, we can calculate $\mathbf{R}$ from the 2-D inverse Fourier transform of $\mathbf{ I }$. For the element in the $i$th row and $h$th column of $\mathbf{R}$, we have 
	\begin{eqnarray}\label{e21}
	\mathbf{R}_{i,h}=\text{vec}(\mathbf{Z}(-i,-h))^{\text{T}}\text{vec}(\mathbf{Z}).
	\end{eqnarray}
    With the aid of the known background information $\mathbf{Y}$, a series of linear combinations of $\text{vec}(\mathbf{X})$ can be obtained from (\ref{e21}). If these combinations are sufficient, then the uniqueness of solutions will be guaranteed. The remaining task is to determine how many linear combinations can be obtained from  to $\mathbf{R}$.\\
	\begin{figure}
		\includegraphics[width=2in,height=2in]{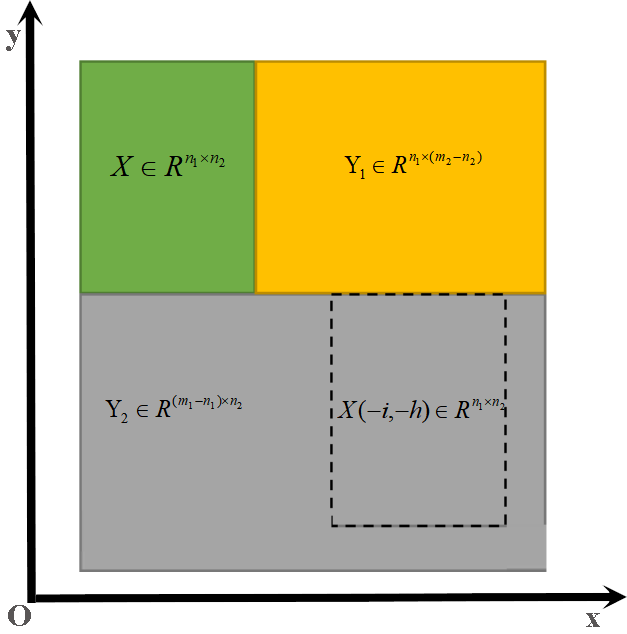}
		\centering
		\caption{Three blocks in $\mathbf{Z}$.}
		\label{f1}
	\end{figure}
	\indent
	For better analysis, we divide $\mathbf{Z}$ into three blocks: $\mathbf{X}$, $\mathbf{Y}_1$ and $\mathbf{Y}_2$ as shown in Figure \ref{f1} using three different colors. $\mathbf{Z}(-i,-h)$ is equivalent to shift $\mathbf{X}$ by $i$ units and $h$ units along two directions namely $\mathbf{X}(-i,-h)$, which is the area surrounded by the dash line. If $\mathbf{ X}$ and $\mathbf{X}(-i,-h)$ do not overlap, the linear combinations of $\text{vec}(\mathbf{X})$ can be obtained from $\mathbf{R}_{i,h}$. Next, we will analyze the numbers of the linear combinations of $\textrm{vec}(\mathbf{X})$, assuming that $\mathbf{X}$ and $\mathbf{X}(-i,-h)$ do not overlap. 
 
	\indent
    \textbf{Firstly, we assumes that $\mathbf{X}(-i,h)$ is completely within the region of $\mathbf{Y}_2$}. Given $h$, let $\mathbf{X}$ only shift along y-axis. Then we have
	\begin{eqnarray}
	\left[
	\begin{array}{c}
	\text{vec}(\mathbf{Z}(n_1,h))_{\mathbf{X}}^{\mathrm{T}}\\
	\text{vec}(\mathbf{Z}(n_1+1,h))_{\mathbf{X}}^{\mathrm{T}}\\
	\vdots\\
	\text{vec}(\mathbf{Z}(k_1,h)_{\mathbf{X}}^{\mathrm{T}}\\
	\end{array}
	\right]
	\text{vec}(\mathbf{X})+\left[
	\begin{array}{c}
	\text{vec}(\mathbf{Z}(k_1,-h)_{\mathbf{X}})^{\mathrm{T}}\\
	\text{vec}(\mathbf{Z}(k_1-1,-h))_{\mathbf{X}}^{\mathrm{T}}\\
	\vdots\\
	\text{vec}(\mathbf{Z}(n_1,-h))_{\mathbf{X}}^{\mathrm{T}}\\
	\end{array}
	\right]
	\text{vec}(\mathbf{X})=\mathbf{R}_1,\label{e22}
	\end{eqnarray}
	where $\text{vec}(\mathbf{Z}(i,h))_{\mathbf{X}}$ are those elements constrained in the region of $\mathbf{X}$. $\mathbf{R}_1$ is determined by $\mathbf{R}_{i,h}$ and $\mathbf{Y}$ ,where $i=n_1,\cdots,k_1$, and  $h=0,\cdots,n_2+k_2-1$. Sliding $\mathbf{X}$ in the region of $\mathbf{Y}_2$ yields a total of $(k_1-n_1+1)(n_2+k_2)$ different linear combinations of $\mathbf{X}$. 
 
	\indent
	\textbf{Secondly, we assumes that $\mathbf{X}(-i,-h)$ overlaps with $\mathbf{Y_1}$}. Given $i$, $\mathbf{X}$ is shifted along x-axis. Similarly, we have
	\begin{eqnarray}
	\left[
	\begin{array}{c}
	\text{vec}(\mathbf{Z}(i,n_2))_{\mathbf{X}}^{\mathrm{T}}\\
	\text{vec}(\mathbf{Z}(i,n_2+1))_{\mathbf{X}}^{\mathrm{T}}\\
	\vdots\\
	\text{vec}(\mathbf{Z}(i,k_2)_{\mathbf{X}}^{\mathrm{T}}\\
	\end{array}
	\right]
	\text{vec}(\mathbf{X})+\left[
	\begin{array}{c}
	\text{vec}(\mathbf{Z}(-i,-n_2)_{\mathbf{X}}^{\mathrm{T}}\\
	\text{vec}(\mathbf{Z}(-i,-n_2-1))_{\mathbf{X}}^{\mathrm{T}}\\
	\vdots\\
	\text{vec}(\mathbf{Z}(-i,-k_2))_{\mathbf{X}}^{\mathrm{T}}\\
	\end{array}
	\right]
	\text{vec}(\mathbf{X})=\mathbf{R}_2, \label{e24}
	\end{eqnarray}
	where $-n_1+1\leq i\leq n_1-1$, or $\mathbf{X}$ is completely in the region of  $\mathbf{Y}_2$. As a result, the total number of linear combinations of $\text{vec}(\mathbf{X})$ in the second situation is $(k_2-n_2+1)(2n_1-1)$.
	
	To conclude, by combining equations \eqref{e22} and \eqref{e24} together, we obtain the following equations: 
	\begin{eqnarray}\label{equal}
	\mathbf{M}\text{vec}(\mathbf{X})=\mathbf{R}_3.
	\end{eqnarray}
Here, $\mathbf{M}\in\mathbb{R}^{\left((k_1-n_1+1)(n_2+k_2)+(k_2-n_2+1)(2n_1-1)\right)\times n_1n_2}$ represents combinations of all the matrices on the left side of the equations. Meanwhile, $\mathbf{R}_{3}$ is the vector stacked by the vectors on the right side of the equations. 

It is worth noticing that, due to symmetry, at most half of $(k_1-n_1+1)(n_2+k_2)+(k_2-n_2+1)(2n_1-1)$ equations will repeat. Since $\mathbf{Y}_{i,h}\overset{i.i.d}{\sim}\mathcal{N}(\mu, \sigma^2)$, $\sigma>0$, if $(k_1-n_1+1)(n_2+k_2)+(k_2-n_2+1)(2n_1-1)\geq 2n_1n_2$, then $\mathbf{M}$ is full column rank with probability 1. We only need to prove that $\text{rank}(\mathbf{M})\geq n_1n_2$ to obtain this conclusion above, because $\text{rank}(\mathbf{M})\leq n_1n_2$, .
 
Due to symmetry, we can construct a sub-matrix $\mathbf{M}_1\in\mathbb{R}^{n_1n_2\times n_1n_2}$ from $\mathbf{M}$ where no two rows are the same. This is because $(k_1-n_1+1)(n_2+k_2)+(k_2-n_2+1)(2n_1-1)\geq 2n_1n_2$. By picking one row from each pair of equivalent rows, we can construct a sub-matrix $\mathbf{M}_1$ with a number of rows greater than $n_1n_2$.
	
Because the determinant of $\mathbf{M}_1$, denoted by $|\mathbf{M}_1|$, is a multivariate polynomial with variables $\mathbf{Y}_{i,h}\overset{i.i.d}{\sim}\mathcal{N}(\mu,\sigma^2)$, the set of roots of $|\mathbf{M}_1|$ is defined as $\mathbf{T}$. Without loss of generality, we assume that there are $d$ independent variables $\textrm{Y}_{1},\textrm{Y}_2,\cdots,\textrm{Y}_d$ in $\mathbf{M}_1$.

Using Lemma 4.1 in \cite{multi}, which states that the roots of the finite real multivariate polynomial in $d$ variables have a measure of zero in $\mathbb{R}^{d}$, we can conclude that $\mathbf{T}$ has a measure of $0$ in $\mathbb{R}^d$. Define the probability density function of each random variable $\textrm{Y}_i$ as $\rho(\textrm{Y}_i)$,  then the probability that 
	\begin{eqnarray*}
\mathbb{P}\left((\textrm{Y}_1,\textrm{Y}_2,\cdots,\textrm{Y}_d)\in\mathbf{T}\right)&=&\int_{(\textrm{Y}_1,\textrm{Y}_2,\cdots,\textrm{Y}_d)\in\mathbf{T}}\rho(\textrm{Y}_1)\rho(\textrm{Y}_2)\cdots\rho(\textrm{Y}_d)\textrm{d}\textrm{Y}_1\textrm{d}\textrm{Y}_2\cdots\textrm{d}\textrm{Y}_d\\
&\leq&\int_{(\textrm{Y}_1,\textrm{Y}_2,\cdots,\textrm{Y}_d)\in\mathbf{T}}d\textrm{Y}_1\textrm{d}\textrm{Y}_2\cdots\textrm{d}\textrm{Y}_d\\
&=&0.
	\end{eqnarray*}
The probability that $|\mathbf{M}_1|$ is nonzero is 1, meaning that $\text{rank}(\mathbf{M})\geq\text{rank}(\mathbf{M}_1)=n_1n_2$. Therefore, $\text{rank}(\mathbf{M})=n_1n_2$, and the signal of interest $\mathbf{X}$ can be uniquely determined.
\end{proof}

Next, we will generalize the result into higher dimensions.
\begin{lemma}\label{my}
Consider the d-dimensional problem described by, 
\begin{eqnarray}
\begin{array} { c } 
\mathrm{ Find }~\mathbf { Z } \\ { \text { s.t. } \mathbf { I } = \left| \mathcal{ F}(\mathbf { Z })  \right| ^ { 2 } } \\
{\mathbf { Z }_{i_1,i_2,\cdots,i_d} = \left\{ \begin{array} { l l } { \mathbf { X }_{i_1,i_2,\cdots,i_d}, } & { ( i_1,i_2,\cdots,i_d) \in\bm{\Omega }} \\ { \mathbf { Y }_{i_1,i_2,\cdots,i_d} , } & { \text { otherwise } } \end{array} \right.,}
\end{array}\label{arbitary}
\end{eqnarray}
where $\mathcal{F}(\cdot)$ is the $d$ dimensional discrete Fourier transform, $\mathbf {I }\in\mathbb{R}^{(n_1+k_1)\times(n_2+k_2)\cdots\times(n_d+k_d)}$, $\mathbf { Z }\in\mathbb{R}^{(n_1+k_1)\times(n_2+k_2)\cdots\times(n_d+k_d)}$, $\mathbf{Y}\in\mathbb{R}^{(n_1+k_1)\times(n_2+k_2)\times\cdots\times(n_d+ k_d)}$ is the known background with $\mathbf{Y}_{i_1,i_2,\cdots,i_d} \overset{i.i.d}{\sim}\mathcal{N}(\mu,\sigma^2)$ where $(i_1,i_2,\cdots,i_d)\notin\bm{\Omega}$, else $\mathbf{Y}_{i_1,i_2,\cdots,i_d}=0$. If $$\prod_{i=1}^d (n_i+k_i)\geq2\prod_{i=1}^d n_i+\prod_{i=1}^d(2n_i-1),$$ then the solution of \eqref{arbitary} is unique. In particular, when $n_i\equiv n$ and $k_i\equiv k,i=1,2\cdots,d$, then $k\geq(2^{\frac{d+1}{d}}-1)n$ is sufficient to guarantee the unique solution.
\end{lemma}
\begin{proof}[Proof of Lemma 2.4]
   The sketch of the proof is similar to that of Theorem \ref{t9}. The only difference is in the method used to count the number of linear equations of vec$(\mathbf{X})$.
    
    For \eqref{arbitary}, $\prod_{i=1}^{d}(n_i+k_i)$ equations of vec$(\mathbf{X})$ will be obtained from $\mathbf{I}$ by using the Wiener-Khinchin theorem and the definition of the high dimensional circular auto-correlation function below 
    \begin{eqnarray*}
    \mathbf{R}_{i_1,i_2,\cdots,i_d}=\textrm{vec}(\mathbf{Z}(-i_1,-i_2,\cdots,-i_d))^{\textrm{T}}\textrm{vec}(\mathbf{Z}),0\leq i_h\leq n_h+k_h,h=1,\cdots,d.
    \end{eqnarray*}
    Then, we can find that there are $\prod_{i=1}^{m}(2n_i-1)$ nonlinear equations of vec$(\mathbf{X})$(see Fig.\ref{f1} to consider 2-D as a special case). Thus, there are $\prod_{i=1}^{d}(n_i+k_i)-\prod_{i=1}^{m}(2n_i-1)$ linear combinations of vec$(\mathbf{X})$. Considering the symmetry of the equations, there are at least $\lfloor\frac{\prod_{i=1}^{d}(n_i+k_i)-\prod_{i=1}^{m}(2n_i-1)}{2}\rfloor$ different equations with probability $1$. By utilizing the same ideas in the proof of Theorem \ref{t9}, we can conclude the proof. 
\end{proof}
\remark{In Lemma \ref{my}, we have $m_i\geq2\cdot2^{\frac{1}{d}}n_i,i=1,2\cdots,d$, which is an extra factor of 2 compared to the number of measurements required by the traditional Fourier PR to guarantee a unique solution, namely $m_i\geq2^{\frac{1}{d}}n_i,i=1,2,\cdots,d$ \cite{miao1998phase}. Although \eqref{arbitary} requires more measurements, it can guarantee an exact unique solution. This property cannot be achieved by traditional Fourier PR through sampling alone.}

\subsection{Stability of the solution}
In this subsection, we will demonstrate another advantage of model \eqref{e15}, namely its stability. This can be used to estimate the distance between the iterations and the ground truth in real experiments.
\begin{theorem}
\label{stability}
    For the model \eqref{e15}, we assume that the background information $\mathbf{Y}_{i,h}\overset{i.i.d}{\sim}\mathcal{N}(\mu,\sigma^2)$, $m_i=n_i+k_i,i=1,2$, and $m_1m_2\geq2n_1n_2+(2n_1-1)(2n_2-1)$. If $\mathbf{Y}$ is given, then for any $\mathbf{X}_1\in\mathbb{R}^{n_1\times n_2}$ and $\mathbf{X}_2\in\mathbb{R}^{n_1\times n_2}$, we have
    \begin{eqnarray*}
    \|\mathbf{X}_1-\mathbf{X}_2\|_{\mathrm{F}}\leq\frac{\delta_1\delta_2}{m_1m_2}\|\mathrm{vec}(\mathbf{I}_1-\mathbf{I}_2)\|_{1}.
    \end{eqnarray*}
    Here $\delta_1$ and $\delta_2$ are the largest singular values of $(\mathbf{M}^{\mathrm{T}}\mathbf{M})^{-1},$ and $\mathbf{M}$ respectively.
\end{theorem}
\begin{proof}
 As demonstrated in the proof of Theorem \ref{t9}, $\mathbf{M}$ is of full column rank with probability $1$, then we have
    \begin{eqnarray}
    \left\|\mathbf{X}_1-\mathbf{X}_2\right\|_{\text{F}}&=&\left\|\text{vec}(\mathbf{X}_1)-\text{vec}(\mathbf{X}_2)\right\|_2\nonumber\\
    &=&\left\|(\mathbf{M}^{\text{T}}\mathbf{M})^{-1}\mathbf{M}^{\text{T}}(\mathbf{R}_3^1-\mathbf{R}_3^2)\right\|_2\nonumber\\
    &\leq& \left\|(\mathbf{M}^{\text{T}}\mathbf{M})^{-1}\right\|_2\left\|\mathbf{M}^{\text{T}}\right\|_2\left\|\mathbf{R}_3^1-\mathbf{R}_3^2\right\|_2\nonumber\\
    &\overset{(a)}{\leq}&\frac{1}{m_1m_2} \left\|(\mathbf{M}^{\text{T}}\mathbf{M})^{-1}\right\|_2
    \left\|\mathbf{M}^{\text{T}}\right\|_2\left\|\mathrm{vec}(\mathbf{I}_1-\mathbf{I}_2)\right\|_{1}\label{cm1},
    \end{eqnarray}
    where (a) is derived by the Wiener-Khinchin Theorem.     
\end{proof}
\begin{remark}
 Suppose for all $\mathbf{X}$ are bounded by a constant $C$ namely $\|\mathbf{X}\|_{\mathrm{F}}\leq C$. We can also derive the upper bound
\begin{eqnarray*}
\left\|\mathrm{vec}(\mathbf{I}_1-\mathbf{I}_2)\right\|_{1}&\leq&m_2\left\|\mathbf{I}_1-\mathbf{I}_2\right\|_{1}\leq\sqrt{m_1}m_2\|\mathbf{I}_1-\mathbf{I}_2\|_{\mathrm{F}}\leq2Cm_1^{\frac{3}{2}}m_2^2\|\mathbf{X}_1-\mathbf{X}_2\|_{\mathrm{F}},
\end{eqnarray*}
where the second inequality is derived by the Cauchy-Schwartz inequality. The last inequality is achieved by:
\begin{eqnarray*}
\|\mathbf{I}_1-\mathbf{I}_2\|_{\mathrm{F}}^2&=&\sum_{l=1}^{m_1}\sum_{h=1}^{m_2}(\mathbf{I}_{1_{l,h}}-\mathbf{I}_{2_{l,h}})^2\\
&=&\sum_{l=1}^{m_1}\sum_{h=1}^{m_2}(\sqrt{\mathbf{I}_{1_{l,h}}}+\sqrt{\mathbf{I}_{2_{l,h}}})^2(\sqrt{\mathbf{I}_{1_{l,h}}}-\sqrt{\mathbf{I}_{2_{l,h}}})^2\\
&\leq&\sum_{l=1}^{m_1}\sum_{h=1}^{m_2}2(\mathbf{I}_{1_{l,h}}+\mathbf{I}_{2_{l,h}})(\sqrt{\mathbf{I}_{1_{l,h}}}-\sqrt{\mathbf{I}_{2_{l,h}}})^2\\
&\overset{(a)}{\leq}&4C^2m_1m_2\sum_{l=1}^{m_1}\sum_{h=1}^{m_2}(\sqrt{\mathbf{I}_{1_{l,h}}}-\sqrt{\mathbf{I}_{2_{l,h}}})^2\\
&\overset{(b)}{\leq}& 4C^2m_1^2m_2^2\sum_{l=1}^{m_1}\sum_{h=1}^{m_2}(\mathbf{X}_{1_{l,h}}-\mathbf{X}_{2_{l,h}})^2,
\end{eqnarray*}
where (a) is derived by the Parseval's theorem, and (b) is achieved by triangle inequality and Parseval's theorem.
\end{remark}
\subsection{Robustness of the solution}
Next, we analyze the situation in which the known background information $\mathbf{Y}$ has bias and the intensity-only measurements $\mathbf{I}$ are corrupted by noise. Under these conditions, we estimate the solution $\mathbf{X}^*$ by solving the least square problem:
\begin{eqnarray}
\mathbf{X}^*\in {\arg\min}_{\mathbf{X}\in\mathbb{R}^{n_1\times n_2}} \|\mathbf{M}\text{vec}(\mathbf{X})-\tilde{\mathbf{R}}_3\|_2^2,\label{LS}
\end{eqnarray}
where $\tilde{\mathbf{R}}_3$ is the corruption of $\mathbf{R}_3$.
If $\mathbf{M}$ has full column rank, then $\mathbf{X}^*$ has the closed form: $$\text{vec}(\mathbf{X}^*)=(\mathbf{M}^{\text{T}}\mathbf{M})^{-1}\mathbf{M}^{\text{T}}\tilde{\mathbf{R}}_3.$$

\noindent Next, we provide a bound for the distance between $\mathbf{X}^*$ and the ground truth $\mathbf{X}$. Although we only consider the 2-D case, it can be generalized to higher-dimensional situations accordingly.
\begin{theorem}\label{robustness}
    For model \eqref{e15}, we assume that $\mathbf{Y}_{i,h}\overset{i.i.d}{\sim}\mathcal{N}(\mu,\sigma^2)$, $m_i=n_i+k_i,i=1,2$,  and $m_1m_2\geq2n_1n_2+(2n_1-1)(2n_2-1)$. $\tilde{\mathbf{I}}$ is the intensity-only measurements that are corrupted by bounded random noise $\bm{\varepsilon}_1$ such that $\max(|\bm{\varepsilon}_1|)\leq c_1$. The known background information $\tilde{\mathbf{Y}}$ also contains bounded random bias $\bm{\varepsilon}_2$ with $\max(|\bm{\varepsilon}_2|)\leq c_2$.
With probability $1$, the estimation $\mathbf{X}^*$ calculated by \eqref{LS} and the ground truth $\mathbf{X}$ satisfy the following inequality:
    \begin{eqnarray*}
    \|\mathbf{X}^*-\mathbf{X}\|_{\mathrm{F}}&\leq&\delta_1\delta_2\left(c_1+c_2\left(2\|\mathrm{vec}(\tilde{\mathbf{Y}})\|_{1}+c_2m_1m_2\right)+\frac{\sqrt{C_2(c_2)\|\mathrm{vec}(\tilde{\mathbf{I}})\|_{1}+C_1(c_1,c_2)}}{m_1m_2}\right),
    \end{eqnarray*}
    where $\delta_1$ and $\delta_2$ are the largest singular values of $(\tilde{\mathbf{M}}^{\mathrm{T}}\tilde{\mathbf{M}})^{-1},$ and $\tilde{\mathbf{M}}$ respectively. $\tilde{\mathbf{M}}$ is the corruption of $\mathbf{M}$. $C_1$ and $C_2$ are constants depending on $c_1$ and $c_2$.
\end{theorem}
\begin{proof}
    Utilizing the same techniques as in proving Theorem \ref{t9}, we can demonstrate that $\tilde{\mathbf{M}}$ is also of full column rank with probability $1$ when being corrupted by bounded random noise $\varepsilon_1$.
    Then we can have
    \begin{eqnarray}
    \left\|\mathbf{X}^*-\mathbf{X}\right\|_{\text{F}}&=&\left\|\text{vec}(\mathbf{X}^*)-\text{vec}(\mathbf{X})\right\|_2\nonumber\\
    &=&\left\|(\tilde{\mathbf{M}}^{\text{T}}\tilde{\mathbf{M}})^{-1}\tilde{\mathbf{M}}^{\text{T}}\tilde{\mathbf{R}}_3-\text{vec}(\mathbf{X})\right\|_2\nonumber\\
    &\leq& \left\|(\tilde{\mathbf{M}}^{\text{T}}\tilde{\mathbf{M}})^{-1}\right\|_2\left\|\tilde{\mathbf{M}}^{\text{T}}\left(\tilde{\mathbf{R}}_3-\tilde{\mathbf{M}}\text{vec}(\mathbf{X})\right)\right\|_2\nonumber\\
    &\overset{(a)}{=}& \left\|(\tilde{\mathbf{M}}^{\text{T}}\tilde{\mathbf{M}})^{-1}\right\|_2\left\|\tilde{\mathbf{M}}^{\text{T}}\left(\mathbf{M}\text{vec}(\mathbf{X})-\mathbf{R}_3+\tilde{\mathbf{R}}_3-\tilde{\mathbf{M}}\text{vec}(\mathbf{X})\right) \right\|_2\nonumber\\
    &\leq& \left\|(\tilde{\mathbf{M}}^{\text{T}}\tilde{\mathbf{M}})^{-1}\right\|_2
    \left\|\tilde{\mathbf{M}}^{\text{T}}\right\|_2\left( \left\|\mathbf{M}-\tilde{\mathbf{M}}\right\|_2\left\|\text{vec}(\mathbf{X})\right\|_2+\left\|\tilde{\mathbf{R}}_3-\mathbf{R}_3\right\|_2\right)\label{cm},
    \end{eqnarray}
    where (a) is derived by the \eqref{equal}, and $\tilde{\mathbf{R}}_3$ is the corruption of $\mathbf{R}_3$.     
    Next, we will divide \eqref{cm} into three parts to analyze. 
    
    For $ \left\|\mathbf{M}-\tilde{\mathbf{M}}\right\|_2$, the relationship below can be held
    \begin{eqnarray}\label{noise1}
   \left\|\mathbf{M}-\tilde{\mathbf{M}}\right\|_2&\leq& \left\|\mathbf{M}-\tilde{\mathbf{M}}\right\|_{\text{F}} \nonumber\\
    &\leq&2\sqrt{n_1n_2m_1m_2}c_2.
    \end{eqnarray}
Specifically, define $\mathbf{D}:=\mathbf{M}-\tilde{\mathbf{M}}$, where the total number of elements of $\mathbf{D}$ is less than $n_1n_2m_1m_2$. Then based on the structure of $\mathbf{M}$ and $\tilde{\mathbf{M}}$ , we have 
$$|\mathbf{D}_{i,k}|=|\mathbf{M}_{i,k}-\tilde{\mathbf{M}}_{i,k}|\leq 2c_2,$$
where the inequality is derived by the Cauchy-Schwarz inequality.

Next, for $\left\|\text{vec}(\mathbf{X})\right\|_2$, we can obtain the following relationship using Parseval's theorem and the triangle inequality:
\begin{eqnarray}\label{noise2}
\left\|\text{vec}(\mathbf{X})\right\|_2^2&=&\frac{1}{m_1m_2}\sum_{i=1}^{m_1}\sum_{k=1}^{m_2}(\tilde{\mathbf{I}}_{i,k}-\varepsilon_{1_{i,k}})\nonumber\\
&\leq&\frac{1}{m_1m_2}\left\|\text{vec}(\tilde{\mathbf{I}})\right\|_{1}+c_1.
\end{eqnarray}
Finally, for $\left\|\tilde{\mathbf{R}}_3-\mathbf{R}_3\right\|_2,$ it can be divided into two parts based on the structures of $\mathbf{R}_3$: auto-correlation $\mathbf{R}$ and multiplications of background information. With a similar idea, by using the Parseval's theorem, Wiener-Khinchin theorem and Cauchy-Schwarz inequality, relationship can be held below
\begin{eqnarray}\label{noise3}
\left\|\tilde{\mathbf{R}}_3-\mathbf{R}_3\right\|_2&\leq& m_1m_2\left(c_1+2c_2\|\text{vec}(\tilde{\mathbf{Y}})\|_{1}+c_2^2m_1m_2\right).
\end{eqnarray}
Accordingly, define $\mathbf{D}:=\tilde{\mathbf{R}}_3-\mathbf{R}_3.$ For each element $D_{i}$, it can be divided into two parts namely the difference of the auto-correlation $\tilde{\mathbf{R}}_{i,k}-\mathbf{R}_{i,k}$ and the difference of the multiplications of background information $\sum_{i_1,k_1,i_2,k_2}(\tilde{\mathbf{Y}}_{i_1,k_1}\tilde{\mathbf{Y}}_{i_2,k_2}-\mathbf{Y}_{i_1,k_1}\mathbf{Y}_{i_2,k_2})$. 

Because 
\begin{eqnarray*}
&&\left|\sum_{i_1,k_1,i_2,k_2}(\tilde{\mathbf{Y}}_{i_1,k_1}\tilde{\mathbf{Y}}_{i_2,k_2}-\mathbf{Y}_{i_1,k_1}\mathbf{Y}_{i_2,k_2})\right|\\
&\leq&\sum_{i_1,k_1,i_2,k_2}\left|\tilde{\mathbf{Y}}_{i_1,k_1}\tilde{\mathbf{Y}}_{i_2,k_2}-(\tilde{\mathbf{Y}}_{i_1,k_1}-\bm{\varepsilon}_{2_{i_1,k_1}})(\tilde{\mathbf{Y}}_{i_2,k_2}-\bm{\varepsilon}_{2_{i_2,k_2}})\right|\\
&\leq&\sum_{i_1,k_1,i_2,k_2}\left|\tilde{\mathbf{Y}}_{i_1,k_1}\bm{\varepsilon}_{2_{i_2,k_2}}+\tilde{\mathbf{Y}}_{i_2,k_2}\bm{\varepsilon}_{2_{i_1,k_1}}-\bm{\varepsilon}_{2_{i_1,k_1}}\bm{\varepsilon}_{2_{i_2,k_2}}\right|\\
&\leq&\sum_{i_1,k_1,i_2,k_2}c_2\left|\tilde{\mathbf{Y}}_{i_1,k_1}\right|+c_2\left|\tilde{\mathbf{Y}}_{i_2,k_2}\right|+c_2^2\\
&\leq&2c_2\|\text{vec}(\tilde{\mathbf{Y}})\|_{1}+c_2^2m_1m_2.
\end{eqnarray*}
Then, we have
\begin{eqnarray*}
\|\mathbf{D}\|_2&\leq&\|\text{vec}(\tilde{\mathbf{R}}-\mathbf{R})\|_{1}+m_1m_2
\left(2c_2\|\text{vec}(\tilde{\mathbf{Y}})\|_{1}+c_2^2m_1m_2\right)\\
&\leq&\sqrt{m_1m_2\|\tilde{\mathbf{R}}-\mathbf{R}\|_{\text{F}}^2}+m_1m_2
\left(2c_2\|\text{vec}(\tilde{\mathbf{Y}})\|_{1}+c_2^2m_1m_2\right)\\
&\leq&\|\tilde{\mathbf{I}}-\mathbf{I}\|_{\text{F}}+m_1m_2
\left(2c_2\|\text{vec}(\tilde{\mathbf{Y}})\|_{1}+c_2^2m_1m_2\right)\\
&\leq&\|\bm{\varepsilon}_1\|_{\text{F}}+m_1m_2
\left(2c_2\|\text{vec}(\tilde{\mathbf{Y}})\|_{1}+c_2^2m_1m_2\right).
\end{eqnarray*}
\noindent Thus, combining \eqref{noise1},\eqref{noise2}, and \eqref{noise3}, we can finish the conclusion.
\end{proof}
\remark{In Theorem \ref{robustness}, the right hand side of the inequality is composed of three parts. The first part is the influence caused by measurement noise $c_1$. The second part is the influence caused by the bias of the background information $C_2(c_2)\|\mathrm{vec}(\tilde{\mathbf{I}})\|_{1}$, and $c_2\left(2\|\mathrm{vec}(\tilde{\mathbf{Y}}\|_{1})+c_2m_1m_2\right)$. The third part is the coupling between the measurement error and bias of the background information namely $C_1(c_1,c_2)$. When $c_1=0$, that is to say there is no measurement noise, then we have $$ \|\mathbf{X}^*-\mathbf{X}\|_{\mathrm{F}}\leq   \delta_1\delta_2\left(c_2\left(2\|\mathrm{vec}(\tilde{\mathbf{Y}})\|_{1}+c_2m_1m_2\right)+\frac{\sqrt{C_2(c_2)\|\mathrm{vec}(\tilde{\mathbf{I}})\|_{1}}}{m_1m_2}\right).$$
When $c_2=0$ namely there is no bias of the background information, we can achieve the relationship below
$$ \|\mathbf{X}^*-\mathbf{X}\|_{\mathrm{F}}\leq c_1\delta_1\delta_2.$$
}
\section{Methods to find the ground truth}
\indent
In this section, we assume that the length of the background information is sufficient to ensure that (\ref{bkmodel}) has a unique solution. We will introduce two methods based on the Douglas-Rachford method to find the ground truth. These methods include the convex and non-convex types. For simplicity, we will only discuss the 1-D formulation of the problem, but the situation in higher dimensions can be generalized.
\subsection{Background Douglas Rachford method}
Denoting $\mathbf{z}=[z_1,\cdots,z_{n+k}]^{\text{T}}$, the 1-D Fourier PR model with background information can be formulated as below 
\begin{eqnarray}
\begin{array} { c } { \text { Find } \mathbf { z } } \\[5pt] { \text { s.t. } \left| \mathbf { F } _ { i } ^{\mathrm{H}} \mathbf { z } \right| ^ { 2 }=b_i , i = 1 , \cdots , m } \\[5pt] { z_{n + l} = y_l , l = 1,\cdots,k.} \end{array}\label{bkmodel}
\end{eqnarray} 
Define set $\mathbf{A}$ and $\mathbf{B}$ as below:
\begin{eqnarray*}
	&\mathbf{A}&:=\left\{\mathbf{z}\in\mathbb{R}^{n+k}: \vert\mathbf{F}^{\mathrm{H}}\mathbf{ z}\vert=\mathbf{ b}^{\frac{1}{2}} \right\}\\
	&\mathbf{B}&:=\left\{\mathbf{z}\in\mathbb{R}^{n+k}: z_{n+i}= y_i,~i=1,\cdots, k\right\}.
\end{eqnarray*}
$\mathbf{ A}$ and $\mathbf{ B }$ are closed sets, however $\mathbf{A}$ is not convex.  $\mathbf{ A}\cap\mathbf{B}$ is actually the set of solution in (\ref{bkmodel}). 
 The projection of $\mathbf{z}\in\mathbb{R}^{n+k}$ onto a closed set $\mathbf{C}$ is defined as $$
	\mathbb{ P }_{\mathbf{ C}}(\mathbf{z}) := \arg\min _ { \mathbf{x} \in \mathbf{C} } \left\| \mathbf{x} - \mathbf{z} \right\|_2^2.$$
Notice that,  when $z_i=0,~i\in\{1,\cdots,n+k\}$, $\mathbb{P}_{\mathbf{C}}(\mathbf{z})$ may be not unique.
In the previous work, \cite{Yuan_2019} reformulated (\ref{bkmodel}) as a constrained optimization problem below 
\begin{eqnarray}  { \underset { \mathbf { z } \in \mathbb { R } ^ { n + k } } { \operatorname { minimize } } }~& f ( \mathbf { z } )& = \frac { 1 } { 2 ( n + k ) } \sum _ { i = 1 } ^ { n + k } \left( \left| \mathbf { F } _ { i } ^{\mathrm{H}} \mathbf { z } \right| - b _ { i } ^ { \frac { 1 } { 2 } } \right) ^ { 2 } \nonumber  \\ 
\text{s.t.}~&\mathbf { z } \in \mathbf{B} &\label{e27}. \end{eqnarray}
Then, the projected gradient descent(PGD) method is utilized to search for the ground truth. The main procedure in the $p$th step is:
$$\left\{	
   \begin{aligned}
	&\tilde{\mathbf{z}} ^ { p } = \mathbf { z } ^ { p-1 } - \lambda \partial f \left( \mathbf { z } ^ { p-1 } \right)&\\
	&\mathbf { z } ^{ p } = \mathbb { P } _ { \mathbf{B}} \left( \tilde{\mathbf { z }} ^ { p  }  \right)&
	\end{aligned}
	\right.,$$
where $\partial f(\mathbf{z}^{p-1})=\mathbf { z }^{p-1} - \mathbb{P}_{\mathbf{A}}(\mathbf{z}^{p-1})$, and $\lambda$ is the learning rate. When $\lambda$ is set to $1$, it ensures that $f(\mathbf{z}^{p})\leq f(\mathbf{z}^{p-1})$. In this case, the PGD method is equivalent to the alternating minimization method\cite{Fienup1982Phase}. Although \cite{Yuan_2019} provides numerous tests to demonstrate the advantages of the PGD method. However, it is prone to stagnation and demands a large number of background information, nearly $k\geq 6n$, to achieve a higher recovery rate. 

In many literatures, the traditional Fourier PR problem has been solved using the Douglas Rachford(DR) method \cite{Fienup1982Phase}, which has demonstrated better performance than the PGD method. In this work, we apply a reflection-based method called background Douglas-Rachford(BDR) method to solve \eqref{bkmodel}. See Fig.\ref{ff1} for an example. The length of the background information is 10. Both methods have a common initialization of $(-0.82,1.26,0.56)$. We observe that the PGD method easily stagnates, but the BDR method can overcome such traps and converge to the ground truth. The details of the BDR method are shown in Algorithm \ref{A1}.
	\begin{figure}
	\includegraphics[width=1\textwidth]{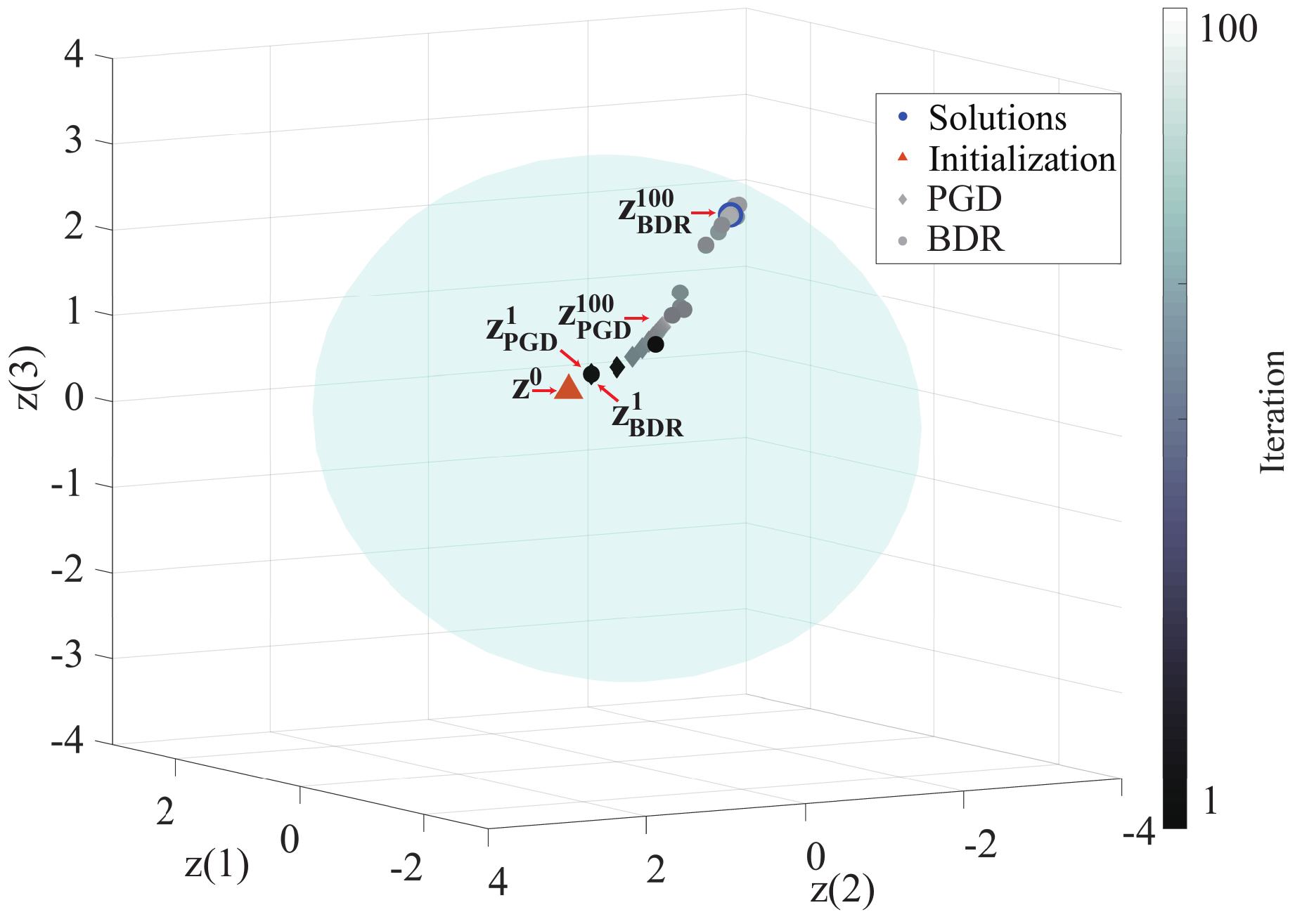}
	\centering
	\caption{The comparison between Projected Gradient Descent(PGD) method and the background Douglas-Rachford (BDR) method, both of which aim to find the ground truth value of $(-1.52,-0.24,2.61)$. The background information has a length of 10, and both methods are initialized with the same values of $(-0.82,1.26,0.56)$. While PGD stagnates easily, the BDR method can escape the trap and converge to the ground truth.}
	\label{ff1}
\end{figure}
\begin{algorithm}[!htb] 
	\renewcommand{\algorithmicrequire}{\textbf{Input:}}
	\renewcommand\algorithmicensure {\textbf{Output:} }
	\caption{ The Background Douglas Rachford method} 
	\label{A1} 
	\begin{algorithmic}[1] 
		\REQUIRE $\{\mathbf{b},\mathbf{y},\varepsilon,\mathrm{T}\}$ ~~\\ 
		$\mathbf{b}$: the intensity only Fourier measurement.\\
		$\mathbf{y}$: the background information.\\
		$\varepsilon$:~~the allowed error bound.\\
		$\text{T}$:~the maximum allowed iteration.\\
		\ENSURE ~~\\ 
		$\overline{\mathbf{x}}$: an estimation of the real signal $\mathbf{x}$.\\
		\vskip 4mm
		\hrule
		\vskip 2mm
	\end{algorithmic}
	$\mathbf{Initialization:}$
	\begin{algorithmic}[1]
		\STATE $\mathbf{z}^0=\mathbb{P}_{\mathbf{ B}}(\frac{1}{n+k}\mathbf{F}\cdot\mathbf{b}^{\frac{1}{2}})$,  $\mathbf{B}=\{\mathbf{z}:z_{n+i}=y_i,~i=1,\cdots,k\}$.\\
		\STATE $p=1$.
	\end{algorithmic}
	$\mathbf{General~step}$
	\begin{algorithmic}[1]
		\STATE  $\tilde{\mathbf{z}}^{p-1}=\mathbb{P}_{\mathbf{ A}}(\mathbf{z}^{p-1})$, where $\mathbf{A}=\{\mathbf{z}:|\mathbf{F}^{\mathrm{H}}\mathbf{ z}|=\mathbf{ b}^{\frac{1}{2}} \}$\\
		\STATE $z^{p}_t=\tilde{z}^{p-1}_t, t=1,\cdots,n$,\\ $z^{p}_{n+t}=z^{p-1}_{n+t}-\tilde{z}^{p-1}_{n+t}$+$y_t,t=1,\cdots,k$.\\
		\IF{$\big|\big|\mathbf{z}^{p}-\mathbf{ z}^{p-1}\big|\big|_2\leq\varepsilon$ \textbf{or} $p=\textrm{T}+1$}
		\STATE $\overline{\mathbf{z}}=\mathbf{z}^{p}$.\\
				\STATE $\overline{\mathbf{x}}=[\overline{z}_1,\overline{z}_2,\cdots,\overline{z}_n]^{\mathrm{T}}$. 
		\STATE Break.
		\ENDIF
		\STATE $p = p+1$
	\end{algorithmic}
\end{algorithm}
\indent

As discussed above, the most significant difference between the PGD and BDR methods is the general step in Algorithm \ref{A1}. For the PGD method, this involves creating a sequence $\{\mathbf{z}^p\}$ by applying alternative projection to approximate $\mathbf{A}\cap\mathbf{B}$:
\begin{eqnarray*}
\mathbf{z}^p = \mathbb{P}_{\mathbf{B}}\mathbb{P}_{\mathbf{A}}(\mathbf{z}^{p-1}).
\end{eqnarray*}
Because $\mathbb{P}_{\mathbf{B}}\mathbb{P}_{\mathbf{ A}}$ is not a strict contraction operator, the PGD method is liable to converge to a fixed point which is usually not the ground truth. For the BDR method, it defines an operator $\mathbb{T}(\cdot)$, and
\begin{eqnarray}\label{e28}
\mathbf{z}^p =\mathbb{T}(\mathbf{z}^{p-1})=\frac{1}{2} (\mathbb{R}_{\mathbf{B}}\mathbb{R}_{\mathbf{A}}+\mathbb{I})(\mathbf{z}^{p-1}),
\end{eqnarray}
where $\mathbb{R}_{\mathbf{A}}(\cdot)$ and $\mathbb{R}_{\mathbf{B}}(\cdot)$ are the reflectors namely
\begin{eqnarray*}
	\mathbb{R}_{\mathbf{A}}(\cdot)=2\mathbb{P}_{\mathbf{A}}(\cdot)-\mathbb{I}(\cdot),~~~~~\mathbb{R}_{\mathbf{B}}(\cdot)=2\mathbb{P}_{\mathbf{B}}(\cdot)-\mathbb{I}(\cdot),~~~\textrm{where}~\mathbb{I}(\cdot)~\text{is the unit projector}.
\end{eqnarray*}
By reformulating (\ref{e28}) as below
\begin{eqnarray*}
	\mathbf{z}^p &=& \frac{1}{2}\big((2\mathbb{P}_{\mathbf{B}}-\mathbb{I})(2\mathbb{P}_{\mathbf{A}}-\mathbb{I})+\mathbb{I}\big)(\mathbf{z}^{p-1})\\
	&=&\big(\mathbb{ P }_{\mathbf{ B}}(2\mathbb{ P }_{\mathbf{A}}-\mathbb{I})+(\mathbb{ I}-\mathbb{P}_{\mathbf{A}})\big)(\mathbf{z}^{p-1})\\
	&\overset{(a)}{=}&\bm{1}_{n}\odot\big(2\mathbb{ P }_{\mathbf{A}}(\mathbf{z}^{p-1})-\mathbf{z}^{p-1}\big)+\big(\mathbf{z}^{p-1}-\mathbb{P}_{\mathbf{A}}(\mathbf{z}^{p-1})\big)+\tilde{\mathbf{y}}\\
	&=&\bm{1}_{n}\odot\mathbb{ P }_{\mathbf{A}}(\mathbf{z}^{p-1})+(\bm{1}-\bm{1}_{n})\odot\big(\mathbf{z}^{p-1}-\mathbb{ P }_{\mathbf{A}}(\mathbf{ z}^{p-1})+\tilde{\mathbf{y}}\big),
\end{eqnarray*} 
where $(a)$ is deduced from the definition of $\mathbb{P}_{\mathbf{B}}(\cdot)$, where $\bm{1}_{n}\in\mathbb{R}^{n+k}$ is a vector whose first $n$ elements are $1$, and the rest are $0$, and $\tilde{ \mathbf { y} }=[\bm{0};\mathbf{y}]\in\mathbb{R}^{n+k}$. This allows us to obtain the main step in Algorithm \ref{A1}.

One advantage of Algorithm \ref{A1} is that if we can find the fixed points, we can derive the ground truth from them. This is proven by the following theorem.

\begin{theorem}\label{t12}
If Algorithm \ref{A1} converges to $\overline{\mathbf{z}}$, which is one of the fixed points of $\mathbb{T}$, then $\mathbb{P}_{\mathbf{B}}\overline{\mathbf{z}}$ is a solution to (\ref{bkmodel}).
\end{theorem}
\begin{proof}[Proof of Theorem 3.1]
	If $\overline{\mathbf{z}}$ is one of the stationary points namely $\overline{\mathbf{z}}=\mathbb{T}(\overline{\mathbf{z}})$. Recalling Algorithm \ref{A1}, we have
	\begin{eqnarray*}
		\tilde{ \overline{z}}_t &=&\mathbb{T}(\overline{\mathbf{z}})_t=\overline{z}_t,~1\leq t \leq n\\
		\tilde{ \overline{z} }_t&=&y _{t-n},~n+1\leq t\leq n+k.	
	\end{eqnarray*}
	Because $|\mathbf{F}^{\mathrm{H}}\tilde{ \overline{\mathbf{z}}}|^2=\mathbf{b}$, $\tilde{ \overline{\mathbf{z}}}=\mathbb{P}_{\mathbf{B}}(\overline{\mathbf{z}})$ is one of the solutions of (\ref{bkmodel}). In particular, if $\mathbf{y}$ is sufficient, Theorem \ref{t9} guarantees that the solution is uniquely determined.
\end{proof}

Next, we will prove the local R-linear convergence of Algorithm \ref{A1}. Before proceeding, we will make some assumptions.
\begin{assumption}\label{assumption}~
	\begin{enumerate}
		\item $m\geq 2(n+k)-1$ namely we have oversampled the Fourier amplitude $\mathbf{b}$ with over $2(n+k)-1$ measurements.
		\item $\mathbf{x}\neq\bm{0}$ and the z-transform of  the ground truth $\mathbf{z}=[\mathbf{x};\mathbf{y}]$ has no zeros in the reciprocal pairs.
	\end{enumerate}
\end{assumption}
The z-transform of a vector $\mathbf{x}$ is denoted by $X(z)$. There are no zeros in the reciprocal pairs means that if $X(z_1)=0$, then  $X(1/\overline{z_1})\neq0.$  By changing the values of background information,  the second assumption can be satisfied.

\begin{theorem}\label{mt}
For model (\ref{bkmodel}), define
\begin{eqnarray*}
	&\mathbf{A}:=\{\mathbf{z}\in\mathbb{R}^{n+k}:|\mathbf{F}^{\mathrm{H}}\mathbf{ z}|=\mathbf{ b}^{\frac{1}{2}},\mathbf{b}\in\mathbb{R}^m \}&\\[5pt]
	&\mathbf{B}:=\{\mathbf{z}\in\mathbb{R}^{n+k}:z_{n+i}= y_i,~i=1,\cdots, k\}.&
\end{eqnarray*}
$\mathbf{F}\in\mathbb{C}^{n\times m}$ is an oversampling Fourier matrix, where $m\geq2(n+k)-1$. Suppose there exists $\delta>0$ such that the initialization $\mathbf{z}^0$ is sufficient close to the ground truth $\mathbf{z}\in\mathbf{ A}\cap\mathbf{B}$, specially $$\left\|\mathbf{z}^{0}-\mathbf{z}\right\|_2 \leq \frac{\delta}{2},$$ 
then the sequence $\{\mathbf{z}^p\}$ generated by Algorithm \ref{A1} converges to $\mathbf{z}$ with R-linear rate, i.e., there exists a constant $\gamma\in[0,1)$ such that
		 $$\left\|\mathbf{z}^{p}-\mathbf{z}\right\|_2 \leq\frac{ \left\|\mathbf{z}^{0}-\mathbf{z}\right\|_2(1+\gamma)}{1-\gamma}\gamma^{p}.$$
\end{theorem}

The sketch of the proof is to establish that $\mathbf{ A}$ and $\mathbf{B}$ are super-regular at $\mathbf{z}\in\mathbf{ A}\bigcap\mathbf{B}$, and the system $\{\mathbf{A},\mathbf{B}
\}$ is strongly regular at $\mathbf{z}$. The results can then be obtained using the theories presented in \cite{phan2016linear}. For details, please refer to Appendix \ref{appendixa}.

In this subsection, we develop the BDR method to solve \eqref{bkmodel}. Although it has been proven that the BDR method exhibits local R-linear convergence, the theoretical results presented above are not optimal. First, the convergence of BDR is local which requires the initialization to be close to the ground truth. This is challenge to achieve in practice.  Second, there is a difference between theories and applications. While the BDR method performs well in numerical tests when $m=n+k$, theories demand $m\geq2(n+k)-1$. Additionally, calculating the operator $\mathbb{P}_{\mathbf{A}}$ is challenging when $m\geq2(n+k)-1$ in the BDR method. 

As a result, in the next subsection, we introduce a convex variant of the BDR method called CBDR method. This method guarantees a global convergence rate and is easier to apply.
\subsection{Convex Background Douglas Rachford method}
In this subsection, we reformulate \eqref{bkmodel} and try to solve the problem below,
\begin{eqnarray}
\begin{array} { c } { \text { Find } \mathbf { z } } \\[5pt] { \text { s.t. } \left| \mathbf { F } _ { i } ^{\mathrm{H}} \mathbf { z } \right| ^ { 2 }\leq b_i , i = 1 , \cdots , m } \\[5pt] { z_{n + l} = y_l , l = 1,\cdots,k.} \end{array}\label{conbkmodel}
\end{eqnarray}
\eqref{conbkmodel} is actually a convex problem. Specifically if any $\mathbf{x}^1$ and $\mathbf{x}^2$ are the solutions of \eqref{conbkmodel}, for $0\leq\lambda\leq1$, we have $\lambda x_{n+l}^1+(1-\lambda)x_{n+l}^2= y_{l},l=1,\cdots,k$, and $$\left|\mathbf{F}_i^{\text{H}}\left(\lambda\mathbf{x}^2+(1-\lambda)\mathbf{x}^2\right)\right|^2\leq\left(\lambda\left|\mathbf{F}_i^{\text{H}}\mathbf{x}^1\right|+(1-\lambda)\left|\mathbf{F}_i^{\text{H}}\mathbf{x}^2\right|\right)^2\leq b_i.$$

If we can demonstrate that \eqref{bkmodel} and \eqref{conbkmodel} have the same solution for a sufficiently large length of background information $k$, we can obtain the ground truth of the non-convex problem \eqref{bkmodel} by solving the convex problem \eqref{conbkmodel}. then a batch of efficient algorithms with theoretical convergence guarantees can be applied. Theorem \ref{Convetheo1} can guarantee this property.

\begin{theorem}\label{Convetheo1}
    Assuming $k>\max\left\{\frac{n}{1-\delta}, \frac{8c\|\mathbf{x}\|_2}{\delta\sqrt{\eta}},\frac{16c(\sqrt{\delta}+2)^2\log(\frac{1}{\eta})}{\delta^2},c(n,\delta)\right\}$, then the probability of the solution being the same for both \eqref{bkmodel} and \eqref{conbkmodel} is no less than $1-\eta$. Here $\delta\in(0,1)$, $c$ is a constant and $c(n,\delta)$ is a constant that depends on $n$ and $\delta$. For details on $c(n,\delta)$, refer to \eqref{sq1}.
\end{theorem}

\begin{remark}
   In \cite{huang2022strong}, it is proven that the natural least squares formulation for affine phase retrieval is strongly convex on the entire space, provided that the measurements are complex Gaussian random vectors and the number of measurements is sufficient. 
   
   However, the main idea to convexify the problem in this subsection is to reformulate $\mathbf{A}:=\left\{\mathbf{z}\in\mathbb{R}^{n+k}: \vert\mathbf{F}^{\mathrm{H}}\mathbf{ z}\vert=\mathbf{ b}^{\frac{1}{2}} \right\}$ to be convex while ensuring that $\mathbf{A}\cap\mathbf{B}$ still belongs to the solutions of \eqref{bkmodel}. 
   To achieve this, we redefine $\mathbf{A}$ as follows: 
   $$\mathbf{A}:=\left\{\mathbf{z}\in\mathbb{R}^{n+k}: \vert\mathbf{F}_i^{\mathrm{H}}\mathbf{ z}\vert\leq b_i^{\frac{1}{2}},i=1,\cdots,n+k\right\}.$$
\end{remark} 
To prove Theorem \ref{Convetheo1}, a matrix $\mathbf{L}$ is constructed. Each column of this matrix is a circular shift of $\mathbf{z}$, which is one of the solutions of \eqref{bkmodel}.
$$\mathbf{L}=\bordermatrix{%
&&&&&&&&\cr
&z_1& z_2 &\cdots&z_n&z_{n+1}&\cdots &z_{n+k-1}&z_{n+k} \cr
&z_2&z_3&\cdots& z_{n+1}&z_{n+2}&\cdots&z_{n+k}&z_1\cr
&\vdots & \vdots & \cdots&\vdots&\vdots &\ddots&\vdots&\vdots \cr
&z_n&z_{n+1}&\cdots&z_{2n}&z_{2n+1}&\cdots&z_{n-2}&z_{n-1}\cr
&z_{n+1}&z_{n+2}&\cdots&z_{2n+1}&z_{2n+2}&\cdots&z_{n-1}&z_{n}\cr
&\vdots & \vdots & \cdots&\vdots&\vdots &\ddots&\vdots&\vdots \cr
&z_{n+k-1}&z_{n+k}&\cdots&z_{n-2}&z_{n-1}&\cdots&z_{n+k-3}&z_{n+k-2}\cr
&z_{n+k}&z_{1}&\cdots&z_{n-1}&z_{n}&\cdots&z_{n+k-2}&z_{n+k-1}\cr
}.
$$
Note that $\mathbf{z}=[\mathbf{x},\mathbf{y}]$, where $\mathbf{x}:=[z_1,z_2,\cdots,z_n]$ and $\mathbf{y}:=[z_{n+1},z_{n+2},\cdots,z_{n+k}]$, and each element of $\mathbf{y}$ satisfies a normal distribution $\mathcal{N}(0,1)$ independently. The theory below proves that $\mathbf{L}$ is non-singular with probability $1$.
\begin{lemma}\label{singularL}
Assume $\mathbf{x}\in\mathbb{R}^n$ is given, as if $k\geq1$,  then $\mathbf{L}\in\mathbf{R}^{(n+k)\times(n+k)}$ is non-singular with probability 1.
\end{lemma}
\begin{proof}
To prove Lemma \ref{singularL}, we can use the similar algebraic tools used in the proof of Theorem \ref{t9} to show that the set of vectors $\mathbf{y}$ which satisfy $|\mathbf{L}| = 0$ has measure $0$ in $\mathbb{R}^k$.
\end{proof}

If $\mathbf{L}$ is non-singular, its columns can span the $\mathbf{R}^{n+k}$ space. Let us define $$\mathbf{e}_1=[1;0;\cdots,0].$$
We can find that $\mathbf{L}\mathbf{e}_1:=\mathbf{z}$, which is one of the solutions of \eqref{conbkmodel}. If there are other solutions of \eqref{conbkmodel} except $\mathbf{z}$, then $\exists \mathbf{h}\in\mathbb{R}^{n+k}$ and $\mathbf{h}\neq\bm{0}$ such that the sub-matrix $\mathbf{L}_1$ of $\mathbf{L}$ satisfies $\mathbf{L}_1\mathbf{h}=\bm{0}$, where
$$\mathbf{L}_1:=\bordermatrix{%
&&&&&&&&\cr
&z_{n+1}&z_{n+2}&\cdots&z_{2n+1}&z_{2n+2}&\cdots&z_{n-1}&z_{n}\cr
&\vdots & \vdots & \cdots&\vdots&\vdots &\ddots&\vdots&\vdots \cr
&z_{n+k-1}&z_{n+k}&\cdots&z_{n-2}&z_{n-1}&\cdots&z_{n+k-3}&z_{n+k-2}\cr
&z_{n+k}&z_{1}&\cdots&z_{n-1}&z_{n}&\cdots&z_{n+k-2}&z_{n+k-1}\cr
},$$
where $\mathbf{L}_1$ is a partial circulant measurement matrix. 

The vector $\mathbf{h}$ mentioned earlier exists inherently because $\text{Rank}(\mathbf{L}_1)\leq k$, which implies that $\text{Ker}(\mathbf{\mathbf{L_1}})\geq n$. Therefore, a space in $\mathbb{R}^{n+k}$ exists, with dimension larger than $n$, that is a subset of $\text{Ker}(\mathbf{L}_1)$.

On the other hand, if $\mathbf{L}\mathbf{h}$ is a solution of \eqref{conbkmodel}, then $|\mathbf{F}^{\text{H}}_i\mathbf{L}\mathbf{h}|^2\leq b_i,i=1,\cdots,n+k$. Note that $\mathbf{L}$ is a circulant matrix, so we have: $$\mathbf{F}\mathbf{L}\mathbf{h}=(\mathbf{F}\mathbf{z})\odot(\mathbf{F}\mathbf{h}),$$
where $\odot$ is the Hadmard product, and $\mathbf{F}$ is the Fourier matrix with the $i$th row being $\mathbf{F}_i^{\text{H}}$. Using the property of Fourier transform, we have: 
$$|\mathbf{F}^{\text{H}}_i\mathbf{L}\mathbf{h}|^2=|\mathbf{F}^{\text{H}}_i\mathbf{z}|^2|\mathbf{F}^{\text{H}}_i\mathbf{h}|^2\leq b_i,i=1,\cdots,n+k.$$
Since $|\mathbf{F}_i^{\text{H}}\mathbf{z}|^2=b_i$, we have $|\mathbf{F}^{\text{H}}_i\mathbf{h}|^2\leq 1,i=1,\cdots,n+k.$

Combining the two situations above, we can conclude that if the sets $\mathbf{C}_1$ and $\mathbf{C}_2$ satisfy $\mathbf{C}_1\cap\mathbf{C}_2=\{\bm{0}\}$, then only one point $\mathbf{z}$ satisfies the constraint in \eqref{conbkmodel}. The sets are defined as follows:
\begin{eqnarray*}
   &\mathbf{C}_1:= \left\{\mathbf{h}|\mathbf{L}_1\mathbf{h}=\bm{0}\right\},\\
&\mathbf{C}_2:=\{\mathbf{h}||\mathbf{F}^{\text{H}}_i\mathbf{h}|^2\leq4,i=1,\cdots,n+k\}.
\end{eqnarray*}

Specifically, suppose that there are two points that satisfy the constraint in \eqref{conbkmodel}. Then, there exist $\mathbf{h}_1\in\mathbb{R}^{n+k}$ and $\mathbf{h}_2\in\mathbb{R}^{n+k}$ such that $\mathbf{L}_1\mathbf{h}_1=\mathbf{L}_1\mathbf{h}_2=\mathbf{y}$, $|\mathbf{F}_i^{\text{H}}\mathbf{h}_1|^2\leq1$, $i=1,\cdots,n+k$, and $|\mathbf{F}_i^{\text{H}}\mathbf{h}_2|^2\leq1$, $i=1,\cdots,n+k$. Thus, we have $\mathbf{L}_1(\mathbf{h}_1-\mathbf{h}_2)=\bm{0}$, and $|\mathbf{F}_i^{\text{H}}(\mathbf{h}_1-\mathbf{h}_2)|^2\leq4$, $i=1,\cdots,n+k$. As a result, $\mathbf{C}_1\cap\mathbf{C}_2=\{\bm{0}\}$ is a sufficient condition to ensure the uniqueness of the solution.

In this paper, we prove that for sufficiently large $k$, and any $\mathbf{h}\in\mathbf{C}_2$, the following inequality holds with probability 1:
\begin{eqnarray}\label{F-RIP}
(c_1(n,k,\mathbf{x})-\delta)\|\mathbf{h}\|_2^2\leq\frac{1}{k}\|\mathbf{L}_1\mathbf{h}\|_2^2\leq(c_2(n,k,\mathbf{x})+\delta)\|\mathbf{h}\|_2^2,
\end{eqnarray}
where $c_1(n,k,\mathbf{x})>\delta$, $\delta\in(0,1)$, $c_1(n,k,\mathbf{x})$ and $c_2(n,k,\mathbf{x})$ are two constants that depend on $n$,$k$ and the ground truth $\mathbf{x}$. This implies that   $\mathbf{C}_1\cap\mathbf{C}_2=\{\bm{0}\}$. As a result, we can conclude that the intersection of $\mathbf{A}$ and $\mathbf{B}$ contains only one point. This verifies that \eqref{bkmodel} and \eqref{conbkmodel} have the same solution.

\begin{remark}
The formulation of \eqref{F-RIP} is similar to the Restricted Isometry Property(RIP) property \cite{candes2005decoding}, which is a standard analysis tool in Compressive Sensing (CS).  However \eqref{conbkmodel} is different. First, in CS, the signal of interest $\mathbf{h}$ is usually assumed to be $s$-sparse, i.e., $\|\mathbf{h}\|_0\leq s$, whereas in this paper, the Fourier spectrum of $\mathbf{h}$ is not large than $1$ at $\omega = \frac{2\pi l}{n+k},l=0,1,\cdots,n+k-1$. By Parseval's theorem, we can also find that $|\mathbf{F}^{\text{H}}_i\mathbf{h}|^2\leq4,i=1,\cdots,n+k$. Second, compared to the partial circulant Gaussian matrix analyzed in the RIP \cite{krahmer2014suprema,haupt2010toeplitz,rauhut2012restricted}, some of the elements in $\mathbf{L}_1$ are not random, so the techniques used to prove the theorem will be different. In summary, we only need to prove the theory below.
\end{remark}
\begin{theorem}[F-RIP]\label{Convetheo2}
   For the phase retrieval with background information model, assume $\mathbf{x}\in\mathbb{R}^n$ is given, if $k>\max\left\{\frac{n}{1-\delta}, \frac{8c\|\mathbf{x}\|_2}{\delta\sqrt{\eta}},\frac{16c(\sqrt{\delta}+2)^2\log(\frac{1}{\eta})}{\delta^2},c(n,\delta)\right\}$, $c$ is a constant and $c(n,\delta)$ is a constant depending on the $n$ and $\delta$, then with probability as least $1-\eta$, the event that 
   \begin{eqnarray}
(c_1(n,k,\mathbf{x})-\delta)\|\mathbf{h}\|_2^2\leq\frac{1}{k}\|\mathbf{L}_1\mathbf{h}\|_2^2\leq(c_2(n,k,\mathbf{x})+\delta)\|\mathbf{h}\|_2^2,~\forall \mathbf{h}\in\mathbf{C}_2
\end{eqnarray}
is held, where $\mathbf{C}_2=\{\mathbf{h}| |\mathbf{F}^{\mathrm{H}}_i\mathbf{h}|^2\leq4,i=1,\cdots,n+k\}$, $c_1(n,k,\mathbf{x})>\delta$, $\delta\in(0,1)$, $c_1(n,k,\mathbf{x})$ and $c_2(n,k,\mathbf{x})$ are two constants that depend on $n$,$k$ and the ground truth $\mathbf{x}$. 
\end{theorem}

To prove Theorem \ref{Convetheo2}, we need to bound the spectral constant of $\mathbf{L}_1$. Let $\mathbf{D}:=\left\{\boldsymbol{h} \in \mathbb{R}^{n+k}:\|\boldsymbol{h}\|_2 \leq 1,\boldsymbol{h}\in\mathbf{C}_2\right\}$, we have
\begin{eqnarray*} 
\delta_s& =&\sup _{\boldsymbol{h} \in\mathbf{D}}\left|\left<\left(\mathbf{V}_{\mathbf{z}}^{\text{H}}\mathbf{V}_{\mathbf{z}}-c_1(n,k)\bm{\mathrm{I}}-\bm{\Phi}^{\text{H}}\bm{\Phi}\right)\mathbf{h},\mathbf{h}\right>\right|=\sup _{\boldsymbol{h} \in\mathbf{D}}\left|\frac{1}{k}\left\|\mathbf{L_1}\mathbf{h}\right\|_2^2-c_1(n,k)\|\boldsymbol{h}\|_2^2-\|\bm{\Phi}\mathbf{h}\|_2^2\right|\\&=&\sup _{\boldsymbol{h} \in\mathbf{D}}\left|\left\|(\mathbf{V}_{\mathbf{z}}\mathbf{h})\right\|_2^2-c_1(n,k)\|\boldsymbol{h}\|_2^2-\|\bm{\Phi}\mathbf{h}\|_2^2\right|
,\end{eqnarray*}
\noindent where $\mathbf{V}_{\mathbf{\mathbf{z}}}:=\frac{1}{\sqrt{k}}\mathbf{P}_{\bm{\Omega}}\mathbf{F}^{-1}\hat{\mathbf{Z}}\mathbf{F}$, and $\bm{\Phi}:=\frac{1}{\sqrt{k}}\mathbf{P}_{\bm{\Omega}}\mathbf{F}^{-1}\hat{\mathbf{Z}}_1\mathbf{F}$, where $\hat{\mathbf{Z}}_1$ is the Fourier transform of  $\mathbf{z}_1=[\mathbf{x};\bm{0}]\in\mathbb{R}^{n+k}.$ The second equality $\frac{1}{k}\|\mathbf{L}_1\mathbf{h}\|^2_2$ is derived by 
\begin{eqnarray*}
        \frac{1}{k}\|\mathbf{L}_1\mathbf{h}\|^2_2&=&\frac{1}{k}\left<\mathbf{L}_1^{\text{T}}\mathbf{L}_1\mathbf{h},\mathbf{h}\right>=\frac{1}{k}\left<\mathbf{L}^{\text{T}}\mathbf{P}_{\bm{\Omega}}^{\text{T}}\mathbf{P}_{\bm{\Omega}}\mathbf{L}\mathbf{h},\mathbf{h}\right>\\
    &=&\frac{1}{k}\left<(\mathbf{F}^{-1}\hat{\mathbf{Z}}\mathbf{F})^{\text{T}}\mathbf{P}_{\bm{\Omega}}^{\text{T}}\mathbf{P}_{\bm{\Omega}}\mathbf{F}^{-1}\hat{\mathbf{Z}}\mathbf{F}\mathbf{h},\mathbf{h}\right>\\ &=&\left<\mathbf{V}_{\mathbf{z}}^{\text{H}}\mathbf{V}_{\mathbf{z}}\mathbf{h},\mathbf{h}\right>.
\end{eqnarray*}

\noindent 
To perform analysis, $\mathbf{V}_{\mathbf{\mathbf{z}}}\mathbf{h}$ is divided into two parts: the first contains random elements $\mathbf{y}$, and the second contains non-random elements $\mathbf{x}$.
\begin{eqnarray*}
\mathbf{V}_{\mathbf{\mathbf{z}}}\mathbf{h}:&=&\frac{1}{\sqrt{k}}\mathbf{P}_{\bm{\Omega}}\mathbf{F}^{-1}\hat{\mathbf{Z}}\mathbf{F}\mathbf{h}=\frac{1}{\sqrt{k}}\mathbf{P}_{\bm{\Omega}}\mathbf{F}^{-1}\hat{\mathbf{H}}\mathbf{F}\mathbf{z}\\
&=&\frac{1}{\sqrt{k}}\mathbf{P}_{\bm{\Omega}}\mathbf{F}^{-1}\hat{\mathbf{H}}\mathbf{F}_1\mathbf{x}+\frac{1}{\sqrt{k}}\mathbf{P}_{\bm{\Omega}}\mathbf{F}^{-1}\hat{\mathbf{H}}\mathbf{F}_2\mathbf{y}.
\end{eqnarray*}

\noindent 
Based on this, we have 
\begin{eqnarray*}
\left\|\mathbf{V}_{\mathbf{z}}\mathbf{h}\right\|_2^2&=&\left\|\frac{1}{\sqrt{k}}\mathbf{P}_{\bm{\Omega}}\mathbf{F}^{-1}\hat{\mathbf{H}}\mathbf{F}_1\mathbf{x}+\frac{1}{\sqrt{k}}\mathbf{P}_{\bm{\Omega}}\mathbf{F}^{-1}\hat{\mathbf{H}}\mathbf{F}_2\mathbf{y}\right\|_2^2\\&=&\|\bm{\Phi}\mathbf{h}\|_2^2+\frac{1}{k}\|\mathbf{P}_{\bm{\Omega}}\mathbf{F}^{-1}\hat{\mathbf{H}}\mathbf{F}_2\mathbf{y}\|^2_2+\frac{1}{k}\left(\left(\mathbf{P}_{\bm{\Omega}}\mathbf{F}^{-1}\hat{\mathbf{H}}\mathbf{F}_1\mathbf{x}\right)^{\text{H}}\mathbf{P}_{\bm{\Omega}}\mathbf{F}^{-1}\hat{\mathbf{H}}\mathbf{F}_2\mathbf{y}\right).
\end{eqnarray*}
Notice that
\begin{eqnarray}
\frac{1}{k}\mathbb{E}\|\mathbf{P}_{\bm{\Omega}}\mathbf{F}^{-1}\hat{\mathbf{H}}\mathbf{F}_2\mathbf{y}\|^2_2&=&\frac{1}{k}\mathbb{E}\|\mathbf{P}_{\bm{\Omega}}\mathbf{F}^{-1}\hat{\mathbf{H}}\mathbf{F}\tilde{\mathbf{y}}\|^2_2\nonumber\\&=&\frac{1}{k} \sum_{\ell=n+1}^{n+k} \mathbb{E} \sum_{i, j=1}^{n+k} \tilde{y}_j \tilde{y}_i h_{\ell -j} h_{\ell - i}=\frac{1}{k} \sum_{\ell=n+1}^{n+k} \sum_{i=n+1}^{n+k}\left|h_{\ell- i}\right|^2=c_1(n,k)\|\boldsymbol{h}\|_2^2,\label{sq}
\end{eqnarray}
where $\tilde{y}_i=0,i=1,2,\cdots,n$, $h_{i}=h_{n+k+i},i=-n-k+1,-n-k,\cdots,0$ and $c_1(n,k)$ is a constant depending on the $k$ and $n$. Specifically, 
\begin{eqnarray*}
\frac{1}{k} \sum_{\ell=n+1}^{n+k} \sum_{i=n+1}^{n+k}\left|h_{\ell- i}\right|^2&=&\frac{1}{k} \sum_{\ell=n+1}^{n+k} \sum_{i=1}^{n+k}\left|h_{\ell- i}\right|^2-\frac{1}{k} \sum_{\ell=n+1}^{n+k} \sum_{i=1}^{n}\left|h_{\ell- i}\right|^2\\
&\geq&\frac{1}{k}\left(k\|\mathbf{h}\|^2_2-n\|\mathbf{h}\|^2_2\right).
\end{eqnarray*}
So $c_1(n,k)\geq\frac{k-n}{k}$. 

As a result, $\delta_s$ can be bounded by the two terms as followings: 
\begin{eqnarray*}
\delta_s&=&\sup _{\boldsymbol{h} \in\mathbf{D}}\left|\frac{1}{k}\|\mathbf{P}_{\bm{\Omega}}\mathbf{F}^{-1}\hat{\mathbf{H}}\mathbf{F}_2\mathbf{y}\|^2_2+\frac{1}{k}\left(\left(\mathbf{P}_{\bm{\Omega}}\mathbf{F}^{-1}\hat{\mathbf{H}}\mathbf{F}_1\mathbf{x}\right)^{\text{H}}\mathbf{P}_{\bm{\Omega}}\mathbf{F}^{-1}\hat{\mathbf{H}}\mathbf{F}_2\mathbf{y}\right) -c_1(n,k)\|\boldsymbol{h}\|_2^2        \right|\\
&\leq&\underbrace{\sup _{\boldsymbol{h} \in\mathbf{D}}\left|\frac{1}{k}\|\mathbf{P}_{\bm{\Omega}}\mathbf{F}^{-1}\hat{\mathbf{H}}\mathbf{F}_2\mathbf{y}\|^2_2-\frac{1}{k}\mathbb{E}\|\mathbf{P}_{\bm{\Omega}}\mathbf{F}^{-1}\hat{\mathbf{H}}\mathbf{F}_2\mathbf{y}\|^2_2\right|}_{(1)}+\underbrace{\sup _{\boldsymbol{h} \in\mathbf{D}}\left|\frac{1}{k}\left(\mathbf{P}_{\bm{\Omega}}\mathbf{F}^{-1}\hat{\mathbf{H}}\mathbf{F}_1\mathbf{x}\right)^{\text{H}}\mathbf{P}_{\bm{\Omega}}\mathbf{F}^{-1}\hat{\mathbf{H}}\mathbf{F}_2\mathbf{y}\right|}_{(2)}.
\end{eqnarray*}
\noindent
The first term (1) in the inequality above can be bounded by using the Chaos process techniques. Define $\mathbf{V}_{\mathbf{h}}:=\frac{1}{\sqrt{k}}\mathbf{P}_{\bm{\Omega}}\mathbf{F}^{-1}\hat{\mathbf{H}}\mathbf{F}_2$, and prove that 
\begin{eqnarray}\label{chaos}
\sup _{\boldsymbol{h} \in\mathbf{D}}\left|\|\mathbf{V}_{\mathbf{h}}\mathbf{y}\|^2_2-\mathbb{E}\|\mathbf{V}_{\mathbf{h}}\mathbf{y}\|^2_2\right|:=\delta_{s_1}\leq\frac{\delta}{2},
\end{eqnarray}
with probability at least $1-\eta$, if $k>\max\left\{\frac{16c(\sqrt{\delta}+2)^2\log(\frac{1}{\eta})}{\delta^2},c(n,\delta)\right\}$, where $\delta\in(0,1)$, and $c$ is a constant. 

For the last term (2), by using the Chebyshev inequality, if $k\geq\frac{8c\|\mathbf{x}\|_2}{\delta\sqrt{\eta}},$
\begin{eqnarray}\label{cheby}
\sup _{\boldsymbol{h} \in\mathbf{D}}\left|\frac{1}{k}\left(\mathbf{P}_{\bm{\Omega}}\mathbf{F}^{-1}\hat{\mathbf{H}}\mathbf{F}_1\mathbf{x}\right)^{\text{H}}\mathbf{P}_{\bm{\Omega}}\mathbf{F}^{-1}\hat{\mathbf{H}}\mathbf{F}_2\mathbf{y}\right|\leq\frac{\delta}{2},
\end{eqnarray}
with probability at least $1-\eta$. 

The proof of \eqref{chaos}, \eqref{cheby} can be found in Appendix. Combing \eqref{chaos}, \eqref{cheby}, and the definition of $\delta_s$, Theorem \ref{mainth} can bound the norm of $\frac{1}{\sqrt{k}}\mathbf{L}_1\mathbf{h}.$  

\begin{theorem}\label{mainth}
 Let the background information $\mathbf{y}=\{y_1,y_2,\cdots,y_n\}$
 be a random vector with entries that satisfy the normal distribution. If, for any $n$, and $\eta,\delta \in(0,1)$,
$$
k>\max\left\{\frac{n}{1-\delta}, \frac{8c\|\mathbf{x}\|_2}{\delta\sqrt{\eta}},\frac{16c(\sqrt{\delta}+2)^2\log(\frac{1}{\eta})}{\delta^2},c(n,\delta)\right\},
$$
 where $c$ is a constant and $c(n,\delta)$ is a constant depending on the $n$ and $\delta$, then with probability at least $1-\eta$, 
 \begin{eqnarray*}
 (c_1(n,k)-\delta)\|\mathbf{h}\|_2^2+\|\mathbf{\Phi}\mathbf{h}\|^2_2\leq\frac{1}{k}\|\mathbf{L}_1\mathbf{h}\|^2_2\leq(c_1(n,k)+\delta)\|\mathbf{h}\|_2^2+\|\mathbf{\Phi}\mathbf{h}\|^2_2, \mathbf{h}\in\mathbf{D},
\end{eqnarray*}
where $c_1(n,k)$ is the constant in \eqref{sq}, and $c_1(n,k)>\delta$. 
\end{theorem}
We can see that $\|\mathbf{L}_1\mathbf{h}\|^2_2=\bm{0}$ if and only if $\mathbf{h}=\bm{0}$. Under this condition, we can prove that $\mathbf{C}_1\cap\mathbf{D}=\{\bm{0}\}$. For $\mathbf{h}\in\mathbf{C}_2\setminus\mathbf{D}$, we have 
$$\left|\mathbf{F}_i^{\text{H}}\left(\frac{\mathbf{h}}{\|\mathbf{h}\|_2}\right)\right|^2\leq\left|\mathbf{F}_i^{\text{H}}\mathbf{h}\right|^2\leq4,~i=1,\cdots,n+k,$$ which implies that $\frac{\mathbf{h}}{\|\mathbf{h}\|_2}\in\mathbf{D}$. Thus, we can generalize Theorem \ref{mainth} from $\mathbf{D}$ to $\mathbf{C}_2$, and Theorem \ref{Convetheo2} is proven. 
\begin{remark}
To maintain that \eqref{bkmodel} and \eqref{conbkmodel} have the same solution, more background information is usually required than to ensure the solution of \eqref{bkmodel} is uniquely defined. This can also be verified through numerical tests in subsection \ref{Pt}, as shown in Theorem \ref{mainth}.
\end{remark}

Now that we have shown that the solution to \eqref{conbkmodel} is equivalent to \eqref{bkmodel}, we can use a series of convex methods to solve \eqref{conbkmodel} and find the ground truth. For consistency, we propose a method called the Background Douglas-Rachford (CBDR) method. Similarly, we define
\begin{eqnarray*}
	&\mathbf{A}:=\{\mathbf{z}\in\mathbb{R}^{n+k}:|\mathbf{F}^{\mathrm{H}}\mathbf{ z}|\leq\mathbf{ b}^{\frac{1}{2}} \}&\\[5pt]
	&\mathbf{B}:=\{\mathbf{z}\in\mathbb{R}^{n+k}:z_{n+i}= y_i,~i=1,\cdots, k\}.&
\end{eqnarray*}
And the details of the CBDR are shown in Algorithm \ref{A2}.
\begin{algorithm}[!htb] 
	\renewcommand{\algorithmicrequire}{\textbf{Input:}}
	\renewcommand\algorithmicensure {\textbf{Output:} }
	\caption{ The Convex Background Douglas Rachford method} 
	\label{A2} 
	\begin{algorithmic}[1] 
		\REQUIRE $\{\mathbf{b},\mathbf{y},\varepsilon,\mathrm{T}\}$ ~~\\ 
		$\mathbf{b}$: the intensity only Fourier measurement.\\
		$\mathbf{y}$: the background information.\\
		$\varepsilon$:~~the allowed error bound.\\
		$\text{T}$:~the maximum allowed iteration.\\
		\ENSURE ~~\\ 
		$\overline{\mathbf{x}}$: an estimation of the real signal $\mathbf{x}$.\\
		\vskip 4mm
		\hrule
		\vskip 2mm
	\end{algorithmic}
	$\mathbf{Initialization:}$
	\begin{algorithmic}[1]
		\STATE $\mathbf{z}^0=\mathbb{P}_{\mathbf{ B}}(\frac{1}{n+k}\mathbf{F}\cdot\mathbf{b}^{\frac{1}{2}})$,  $\mathbf{B}=\{\mathbf{z}:z_{n+i}=y_i,~i=1,\cdots,k\}$.\\
		\STATE $p=1$.
	\end{algorithmic}
	$\mathbf{General~step}$
	\begin{algorithmic}[1]
		\STATE  $\tilde{\mathbf{z}}^{p-1}=\mathbb{P}_{\mathbf{ A}}(\mathbf{z}^{p-1})$, where $\mathbf{A}=\{\mathbf{z}:|\mathbf{F}^{\mathrm{H}}\mathbf{ z}|\leq\mathbf{ b}^{\frac{1}{2}} \}$\\
		\STATE $z^{p}_t=\tilde{z}^{p-1}_t, t=1,\cdots,n$,\\ $z^{p}_{n+t}=z^{p-1}_{n+t}-\tilde{z}^{p-1}_{n+t}$+$y_t,t=1,\cdots,k$.\\
		\IF{$\big|\big|\mathbf{z}^{p}-\mathbf{ z}^{p-1}\big|\big|_2\leq\varepsilon$ \textbf{or} $p=\textrm{T}+1$}
		\STATE $\overline{\mathbf{z}}=\mathbf{z}^{p}$.\\
				\STATE $\overline{\mathbf{x}}=[\overline{z}_1,\overline{z}_2,\cdots,\overline{z}_n]^{\text{T}}$ 
		\STATE Break.
		\ENDIF
		\STATE $p = p+1$
	\end{algorithmic}
\end{algorithm}

Compared to Algorithm \ref{A1}, the main difference in Algorithm \ref{A2} is the operator $\mathbb{P}_{\mathbf{A}}(\cdot)$. In Algorithm \ref{A2}, $\tilde{\mathbf{z}}^{p-1}$ is the projection of $\mathbf{z}^{p-1}$ onto $\mathbf{A}$. We define $\hat{\tilde{\mathbf{z}}}^{p-1}$ and $\hat{\mathbf{z}}^{p-1}$ as the Fourier transforms of $\tilde{\mathbf{z}}^{p-1}$ and $\mathbf{z}^{p-1}$, respectively, which have the following property:
\begin{eqnarray}
\tilde{z}^{p-1}_i=\left\{
\begin{array}{c}
     b_i^{\frac{1}{2}}\frac{\tilde{z}^{p-1}_i}{|\tilde{z}^{p-1}_i|},~\text{if}~|\hat{\tilde{z}}^{p-1}_i|\geq b_i^{\frac{1}{2}}\\
     ~~\\
    |\hat{\tilde{z}}^{p-1}_i|\frac{\tilde{z}^{p-1}}{|\tilde{z}^{p-1}|},~\text{if}~|\hat{\tilde{z}}^{p-1}_i|< b_i^{\frac{1}{2}} 
\end{array}
\right..
\end{eqnarray}

\noindent Then, the convergence of the CBDR method can be proved by the Theorem by using the theoretical results in \cite{Bauschke2002Phase}.
\begin{theorem}
Suppose that $\mathbf{A}\cap\mathbf{B}\neq\varnothing$. Then the sequence $\left\{\mathbf{z}^p\right\}$ generated by the Algorithm \ref{A2} has the property that $\|\mathbf{z}^p-\overline{\mathbf{z}}\|_2\longrightarrow0$ as $p$ increases and $\mathbb{P}_\mathbf{B}\left(\left\{\mathbf{z}^p\right\}\right)$ $\rightarrow\mathbb{P}_\mathbf{B}(\overline{\mathbf{z}}) \in \mathbf{A} \cap \mathbf{B}$, where $\overline{\mathbf{z}}$ is one of the fixed points of the CBDR method.
\end{theorem}

When $\mathbf{z}\in\mathbb{R}^n$, we can further restrict the feasible set $\mathbf{A}$ and search for signals of interest by performing two parallel tasks. Specifically, we apply CBDR method on $\{\mathbf{A}_1,\mathbf{B}\}$, and $\{\mathbf{A}_2,\mathbf{B}\}$ simultaneously, where $$\mathbf{A}_1:=\{\mathbf{z}\in\mathbb{R}^{n+k}:|\mathbf{F}^{\mathrm{H}}_i\mathbf{ z}|\leq b^{\frac{1}{2}}_i,i=2,3,\cdots,n+k, \sum_{i=1}^{n+k}z_i=\sqrt{b_1}\},$$
and 
$$\mathbf{A}_2:=\{\mathbf{z}\in\mathbb{R}^{n+k}:|\mathbf{F}^{\mathrm{H}}_i\mathbf{ z}|\leq b^{\frac{1}{2}}_i,i=2,3,\cdots,n+k, \sum_{i=1}^{n+k}z_i=-\sqrt{b_1}\}.$$ The final estimation is the one with the least measurement error.

\section{Numerical test}
In this section, we will compare the BDR method with the PGD method. In all tests, the background information is generated using the Gaussian distribution. The tests were conducted on a Dell desktop with a 2.70 GHz Intel Core i7 processor and 32GB DDR3 memory.

The relative error is defined as
\begin{eqnarray}
e=\frac{||\overline{\mathbf{x}}-\mathbf{x}||_2}{||\mathbf{x}||_2},
\end{eqnarray} 
where $\overline{\mathbf{x}}$ is the estimation for the ground truth $\mathbf{x}$.

The measurement error is defined as $$me=\frac{\left\||\mathbf{F}^{\text{H}}\overline{\mathbf{z}}|^2-\mathbf{b}\right\|_2}{\|\mathbf{b}\|_2},$$ where $\overline{\mathbf{z}}=[\overline{\mathbf{x}};\mathbf{y}]$.
The length of the 1-D signals used in this test is 100. The size of the 2-D images is $256\times 256$. In all numerical tests, the Fourier measurements have no oversampling, namely $m_i=n_i+k_i,i=1,2$.
\subsection{Phase transition of different methods}\label{Pt}
First, we analyze the phase transition where the recovery rate is above $90$\% for the PGD, CBDR and BDR methods. The signal of interest is $\mathbf{x}\overset{\text{i.i.d}}{\sim}\mathcal{N}(\mathbf{0},\mathbf{I})$. The length of the signal ranges from $100$ to $1000$ with an interval of $5$. For each signal, measurements $\mathbf{b}$ are generated by \eqref{bkmodel} with $k/n$ ranging from $1$ to $7$ with an interval of $0.1$. Then, we apply the PGD, CBDR and BDR methods to recover the signal. A relative error below $10^{-5}$ is considered a successful recovery. The test are repeated $100$ times for each $n$ and $k/n$. The results are shown in Fig \ref{sf2}. 

In Fig \ref{sf2}, the figures from left to right represent the results calculated by the PGD, CBDR and BDR methods. The corresponding red line is the phase transition where the recovery rate is above $90$\%. We observe that the BDR method requires the least background information to achieve the highest recovery rate, followed by the CBDR method. This demonstrates that the DR method is more likely to find the ground truth when the background information is sufficient. Additionally, we depict the phase transition of the BDR method where the recovery rate is above $99$\% with a dashed red line. The mean of the dashed red line is $2.76$. These results reflect the superiority of BDR method which can recover the signal of interest when $k/n$ is approximately equal to the theoretical bound of $3$.

Although CBDR method is complete in theory, it still requires $k/n\approx5.5$ for a higher recovery rate in practice. This is because convexity may introduce some spurious points into the feasible solutions, and more background information is necessary to mitigate this issue. However, CBDR's performance is still better than that of the PGD method. The BDR and CBDR methods both use the Douglas-Rachford method as their backbone. Therefore, in the remaining sections, we will focus only on the BDR method, which can achieve better performance.
\begin{figure}
	\centering
	\includegraphics[width=\textwidth]{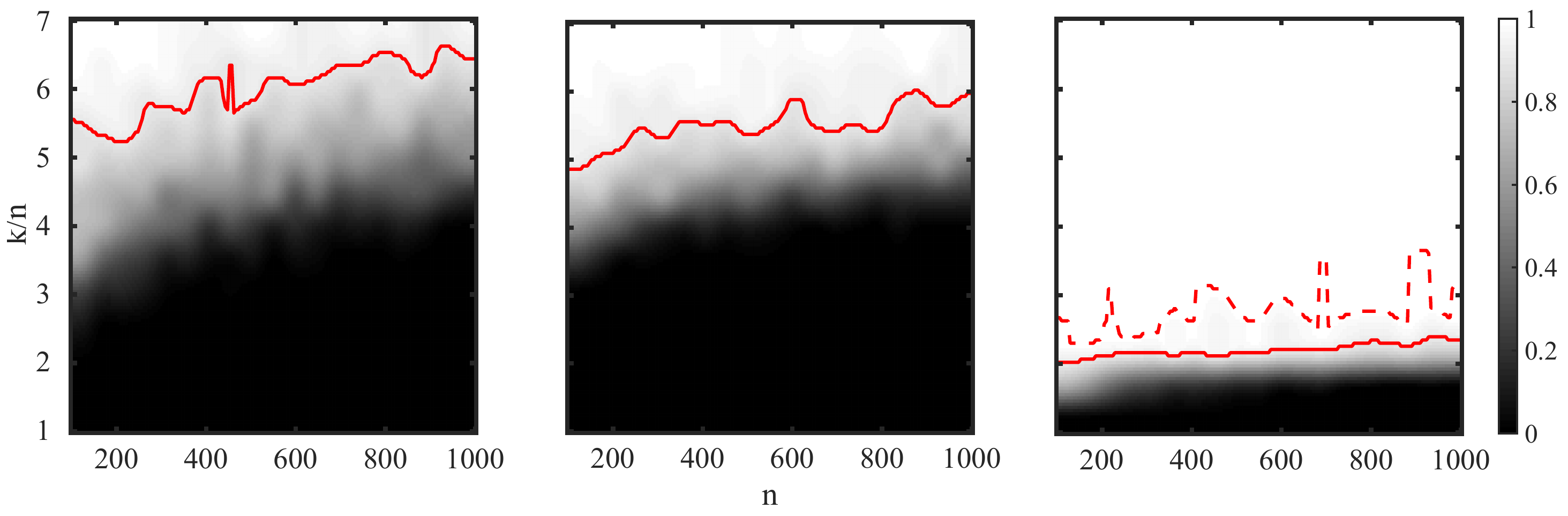}
	\caption{The phase transition of the PGD, CBDR and BDR methods. The figures from left to right represent the results calculated by the PGD, CBDR and BDR methods. The corresponding red line is the phase transition where the recovery rate is above $90$\%. The dashed red line depicts the phase transition where the recovery rate is above $99$\% for the BDR method.}
	\label{sf2}
\end{figure}
\subsection{The effect of the length of signal} 
First, we consider 1-D signals of three different types in the test.
\begin{itemize}
	\item Type 1: $\mathbf{x}\overset{\text{i.i.d}}{\sim}\mathcal{N}(\mathbf{0},\mathbf{I})$
	\item Type 2:
	$f(t)=\text{cos}(39.2\pi t-12\text{sin}2\pi t)+\text{cos}(85.4\pi t+12\text{sin}2\pi t),~t\in(0,1)$
 \item Type 3: the annual mean global surface temperature anomaly which can be seen in Fig.\ref{ano}
\end{itemize}
\begin{figure}
	\centering
	\includegraphics[width=0.5\textwidth]{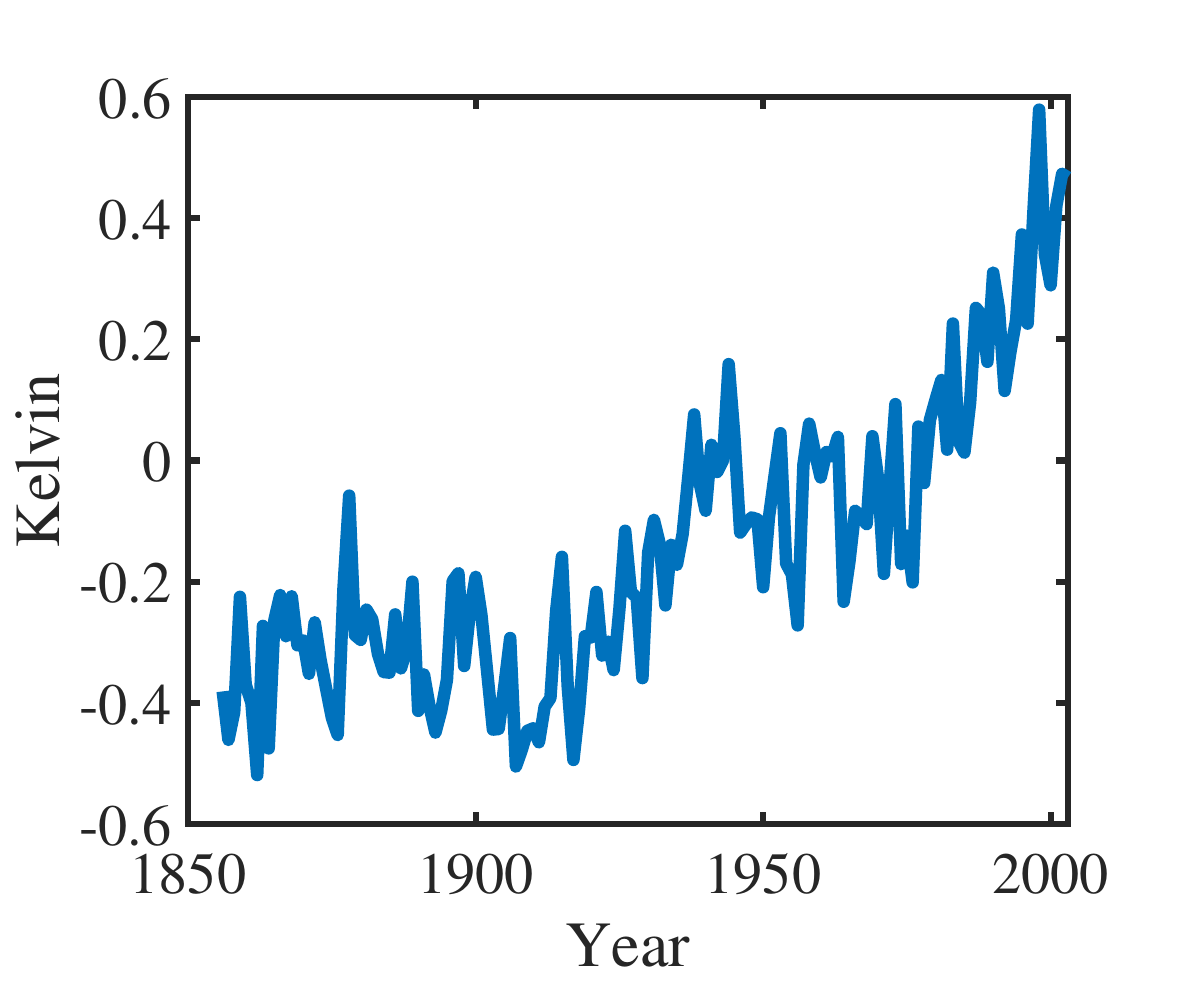}
	\caption{The annual mean global surface temperature anomaly.}
	\label{ano}
\end{figure}
Background information is generated for each signal, and its corresponding Fourier amplitude only spectrum is measured. $k/n$ ranges from $1$ to $7$ with an interval of $0.1$. At each $k/n$, the PR test is applied 100 times with different $\mathbf{y}$. The maximum iteration of both the BDR and PGD methods is 300. Success is judged if the relative error $e$ is below $10^{-5}$. The empirical recovery rate at each $k$ is calculated as the total successful recovery times divided by 100. Figure \ref{f2} shows the results of the test.
\indent
According to Figure \ref{f2}, the BDR method has a higher recovery rate than PGD. Specifically, when $k/n=3$, the BDR method achieves a recovery rate larger than 85\%, which is close to the theoretical bound of $3n-1$, while the recovery rate of PGD remains below 40\%.

In Figure \ref{1Dtest}, we observe that the signal recovered by the BDR method fits the ground truth well, while the PGD method fails to recover the outline of the signal. Moreover, the relative error and measurement error of the BDR method decrease rapidly, while the PGD method stagnates at a local minimum. Even when $k/n=2$, the BDR method still achieves a recovery rate of approximately 60\%, while the PGD method fails to recover any type of signal.

These results fully demonstrate the ability of BDR to alleviate the effect of the stagnation point. It is worth noting that type 1 signals (random signals) are harder to recover than type 2 (determined signals) and type 3 signals (signals with structure). However, in real-life applications, the signal of interest is usually not completely random, which suggests that the BDR method may perform well in practice.

\begin{figure}
	\centering
	\includegraphics[height = 2.4in,width=4.5in]{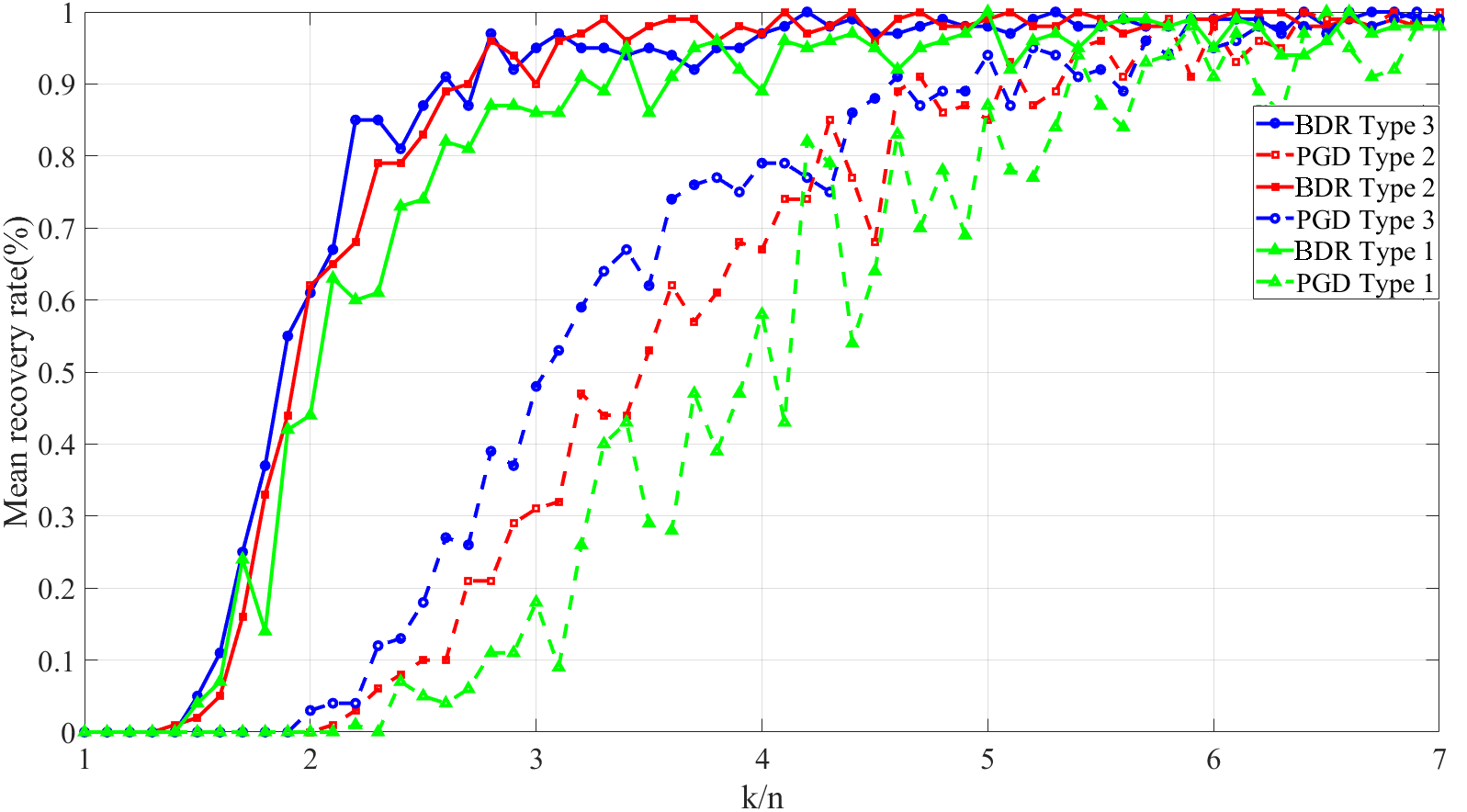}
	\caption{The mean recovery rate of BDR and PGD method. The signal has 3 different types. Type 1 is the random signal, type 2 is the harmonic signal, type 3 is the structure signal.}
	\label{f2}
\end{figure}
\begin{figure}
	\includegraphics[width = 1\textwidth]{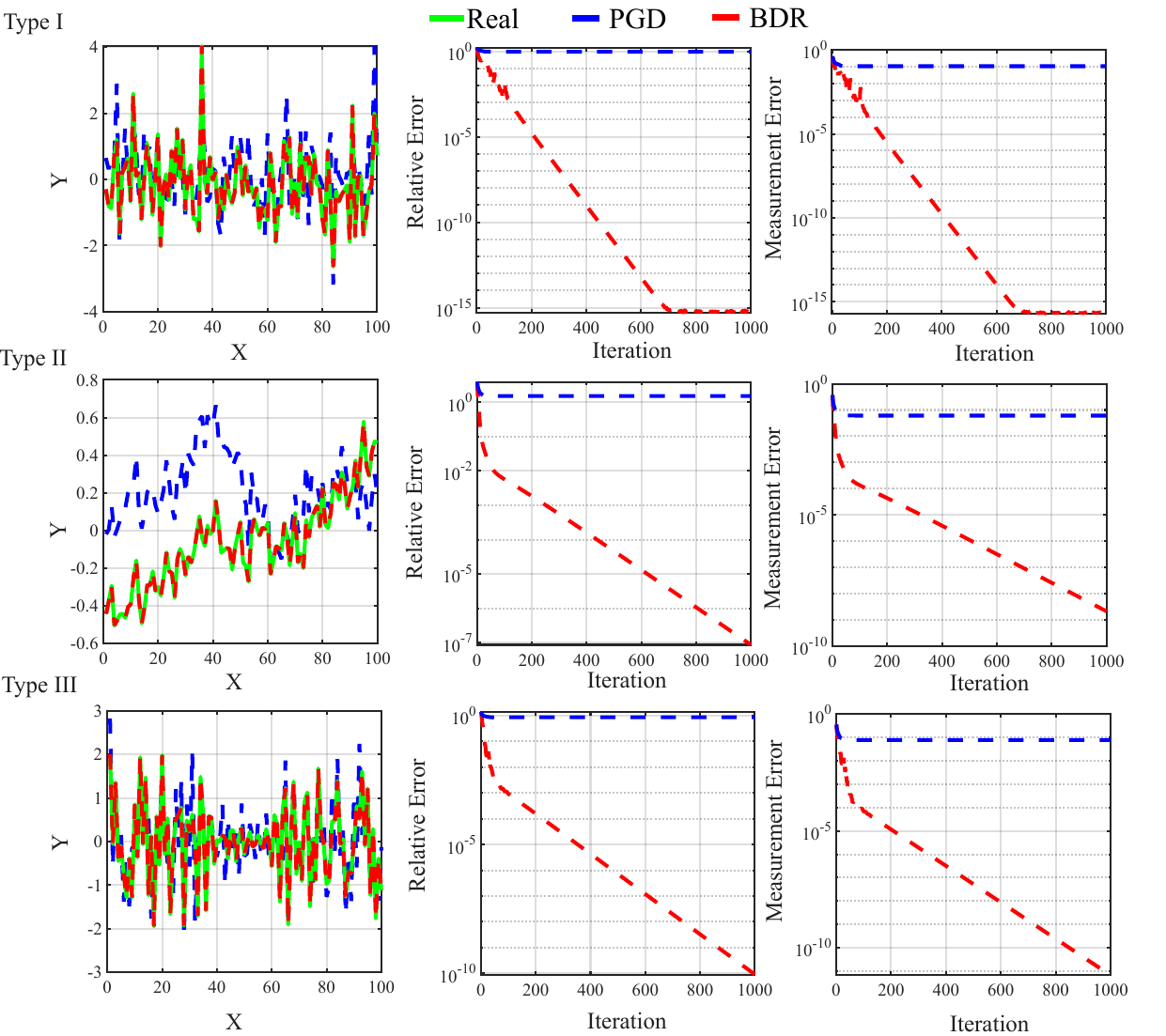}
	\centering
	\caption{Comparisons of BDR and PGD method when $k/n=3$. The first column displays the comparison between the signals recovered by the PGD method and the BDR method. The second column and last column shows the relative error and measurement error of the results in the first column respectively.}
	\label{1Dtest}
\end{figure}
\begin{figure*}[!htb]
	\centering
	\subfigure[Baboon]{\includegraphics[width=.157\textwidth]{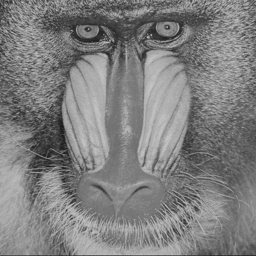}}
	\subfigure[Barbara]{\includegraphics[width=.157\textwidth]{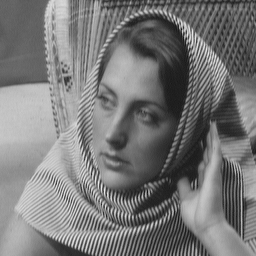}}
	\subfigure[Boat]{\includegraphics[width=.157\textwidth]{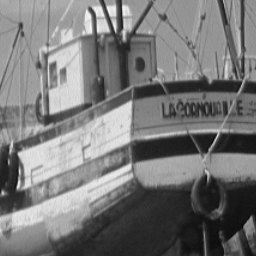}}
	\subfigure[Pepper]{\includegraphics[width=.157\textwidth]{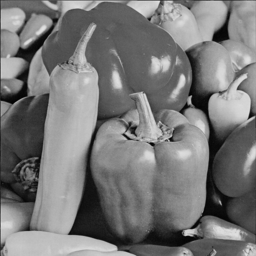}}
	\subfigure[Cameraman]{\includegraphics[width=.157\textwidth]{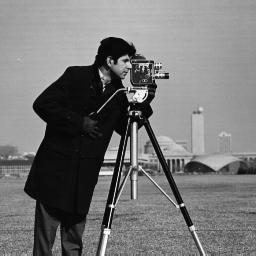}}
	\subfigure[House]{\includegraphics[width=.157\textwidth]{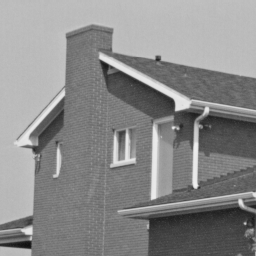}}
	\subfigure[E.~Coli]{\includegraphics[width=.157\textwidth]{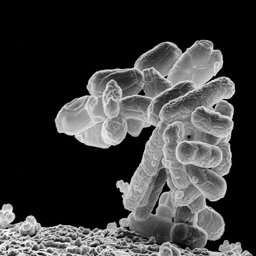}}
	\subfigure[Yeast]{\includegraphics[width=.157\textwidth]{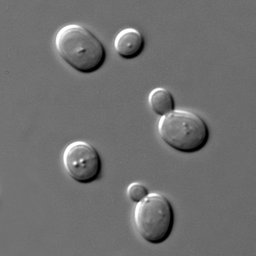}}
	\subfigure[Pollen]{\includegraphics[width=.157\textwidth]{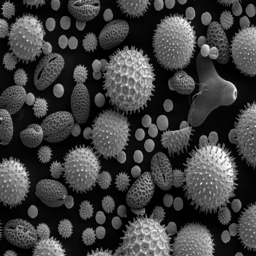}}
	\subfigure[Tadpole Galaxy]{\includegraphics[width=.157\textwidth]{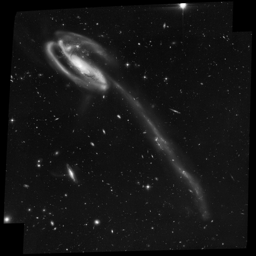}}
	\subfigure[Pillars of Creation]{\includegraphics[width=.157\textwidth]{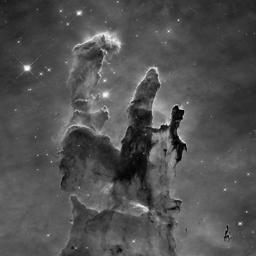}}
	\subfigure[Butterfly Nebula]{\includegraphics[width=.157\textwidth]{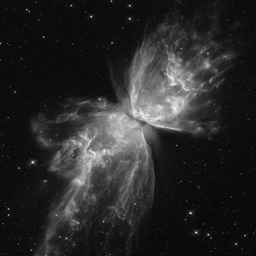}}
	\subfigure[Boston]{\includegraphics[width=.157\textwidth]{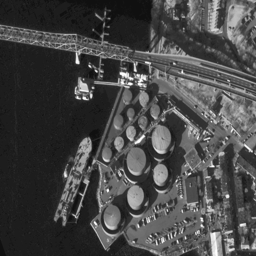}}
	\subfigure[Palm]{\includegraphics[width=.157\textwidth]{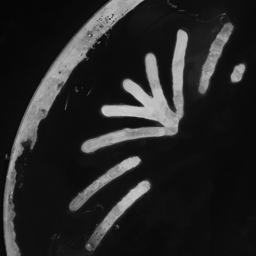}}
	\subfigure[Three Gorges]{\includegraphics[width=.157\textwidth]{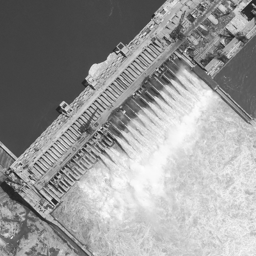}}
	\subfigure[Tokyo]{\includegraphics[width=.157\textwidth]{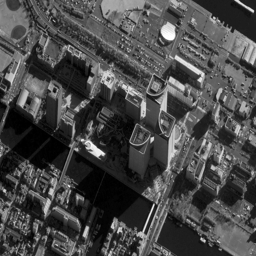}}
	\subfigure[Yokahoma]{\includegraphics[width=.157\textwidth]{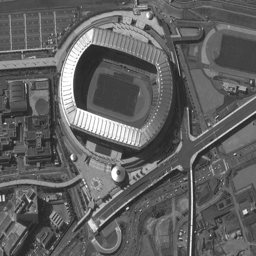}}
	\subfigure[Tucson]{\includegraphics[width=.157\textwidth]{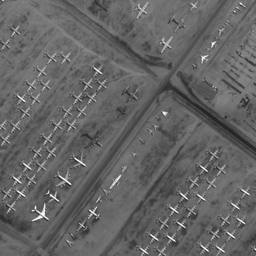}}
	\caption{We attempt to reconstruct 16 test images, including (a-f) natural images, (g-l) astronomy and microbe pictures from Wikipedia under public domain licenses, and (m-r) remote sensing pictures.}
	\label{fig:TestImages}
\end{figure*}

Next, we will test the ability of the BDR method to deal with the 2-D images. The index of peak signal to noise ratio(PSNR) and the structure similarity index(SSIM) are chosen as the criterion to judge the performance of recovery. The time cost by each methods is also recorded. The test pictures are shown in Figure \ref{fig:TestImages}. In each dimension, background information is $k/n$ folds of the object. At each $k/n$ ratio, we record the recovered images' PSNR, SSIM and time under 100 different background information. Each test is applied with 300 iterations. If the relative error is below $10^{-5}$, we stop the iterate in advance. All the images are in the center of the background information. The results are shown in Figure \ref{f3}.

From Figure \ref{f3}, we can find that the medium of the PSNR and SSIM of the BDR method is higher than the PGD method. Especially when $k/n=0.6$, which is mostly smaller than the uniqueness bound required in Theorem \ref{t9}, the medium of the PSNR for the BDR method is above 60dB, but the medium of PSNR for PGD is below 30dB. Figure \ref{f3} indicates that the BDR method can have a good performance when less background information is available. At the same time, we can also find that the variance of the PSNR and SSIM for the BDR method is smaller than PGD method. This also shows that the BDR method can be more numerical stable. With $k/n$ increasing, the time demanded by BDR method is more than PGD method. But when $k/n\leq1$, the time cost by both methods is nearly the same, but BDR method can perform better judging from the PSNR and SSIM.
\begin{figure}[!htb]
	\centering
	\includegraphics[width=0.90\textwidth]{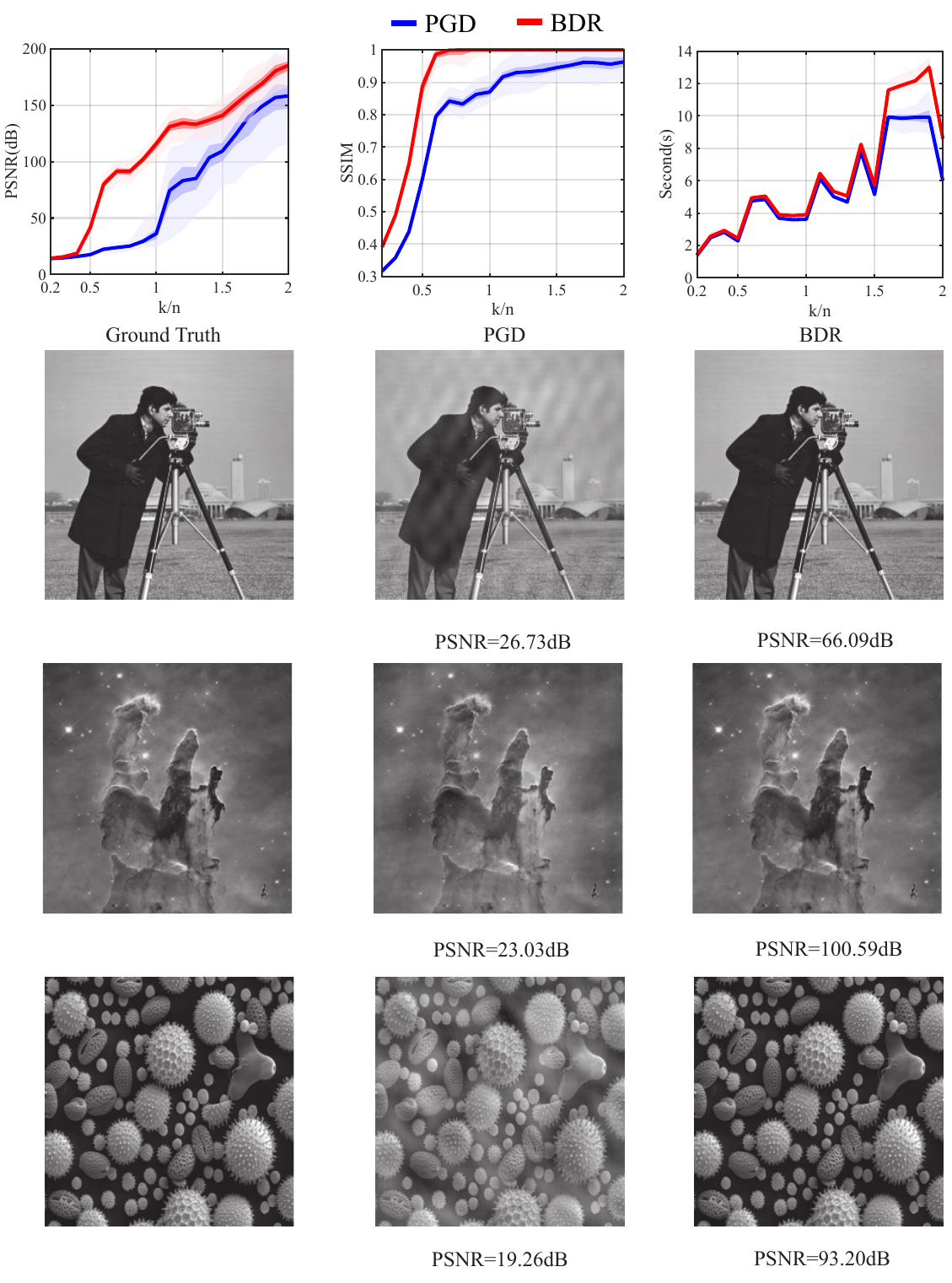}
	\caption{The comparison between PGD and BDR methods for recovering 2D images is presented. In the first row, we compare PSNR, BDR, and time for PGD and BDR. The thick line represents the median over 100 trials, the light area represents the minimum and maximum values, and the darker area indicates the 25th and 75th quantiles. In rows 2-4, we show the results recovered by PGD and BDR when $n/k=0.6$.}
	\label{f3}
\end{figure}
\subsection{The robustness test}
This subsection describes experiments to test the robustness of the BDR algorithm to the locations of the background information and the measurement noise.
  
Firstly, the influence of the locations of the background information is tested. The sample 'Baboon' is to be recovered with $k/n=2$. Seventeen different positions of Baboon are selected, from the left corner to the center in Figure \ref{locbias}(a). At each position, tests are replicated by BDR 10 times. The mean PSNR, mean SSIM, and the mean relative error at each position are recorded (see Figure \ref{locbias}(b), \ref{locbias}(c), and \ref{locbias}(d)). In Figure \ref{locbias}(b), \ref{locbias}(c), and \ref{locbias}(d), the coordinate of each pixel indicates the location of the left-top of the Baboon in the combined picture. Due to the symmetry of positions, only four-ninths parts are considered in Figure \ref{locbias}(a).

As seen from Figure \ref{locbias}(a), similar to the PGD method in \cite{Yuan_2019}, the locations of the background information also affect the performance of the BDR method. Specifically, from the left corner to the center, SSIM and PSNR of the recovered image gradually increase. However, compared to PGD in \cite{Yuan_2019}, the BDR method is less sensitive to the locations of the background information.

Next, experiments are applied to test the robustness of the algorithm to the measurement noise for 2-D pictures. The noise model is given by:
\begin{eqnarray*}
	\sqrt{b}_i = |\mathbf{F}^{\mathrm{H}}_i\mathbf{z}|+\varepsilon_i,i=1,\cdots,m,
\end{eqnarray*} 
where $\varepsilon_i$ is the Gaussian noise.

Five different single-channel pictures were selected as objects, and the BDR and PGD methods were used to recover these images from their noisy measurements. Table \ref{T1} shows the PSNR and SSIM of the recovered images.
The results in Table \ref{T1} demonstrate that the BDR method is less robust to measurement noise than the PGD method, but it requires less background information. To improve its ability to resist noise, we adapted the third step of Algorithm \ref{A1} as follows:
\begin{eqnarray*}
&z^{p}_t=\tilde{z}^{p-1}_t, t=1,\cdots,n,\\
&z^{p}_{n+t}=z^{p-1}_{n+t}-\beta\tilde{z}^{p-1}_{n+t}+y_t,t=1,\cdots,k,
\end{eqnarray*}
Here, $\beta$ is a parameter that controls the projection. This modified algorithm is named BDR1. In the numerical test, we set $\beta=0.9$. Table \ref{T1} shows that this adaptation can improve the robustness of the BDR algorithm and achieve better performance than the PGD method.
\begin{figure}
	\centering
	\includegraphics[width=5in]{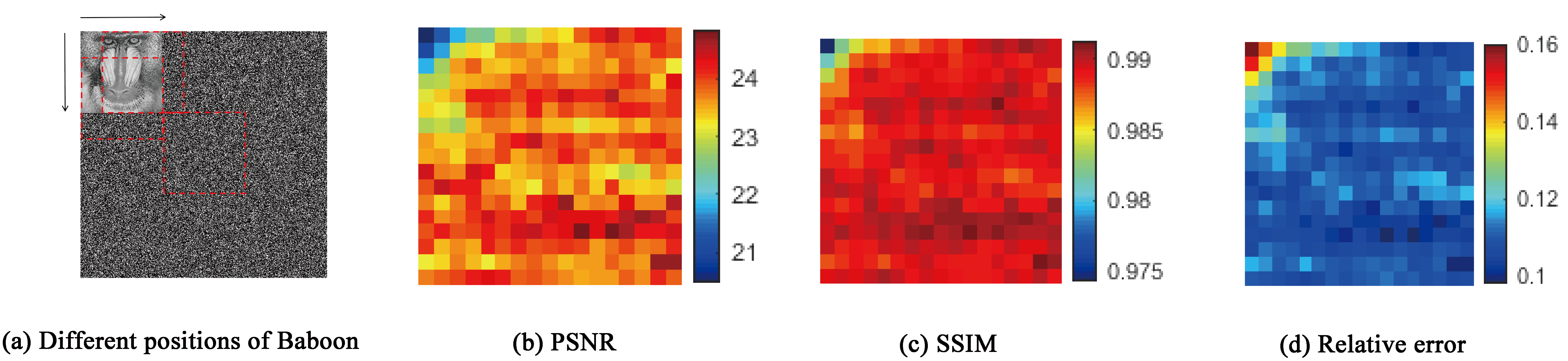}
	\caption{ Pictures recovered under different locations of the background information.
	}\label{locbias}
\end{figure}
\begin{table}
	\centering
	\caption{Results for the robustness test of 2-D Fourier PR with background information, recovering noisy measurements.}
	\vspace{0.01cm}
	\scalebox{0.73}{
		\begin{tabular}{ccccccccccc}
			\hline
			\multicolumn{2}{c}{} &\multicolumn{2}{c}{PSNR}& & \multicolumn{2}{c}{SSIM}& &\multicolumn{2}{c}{Relative error}\\
			\cline{3-4}\cline{6-7}\cline{9-10}
			\multicolumn{2}{c}{}&$\sigma=0.001$ & $\sigma=0.003$ & &$\sigma=0.001$ & $\sigma=0.003$& & $\sigma=0.001$&$\sigma=0.003$ \\
			\hline
			\multirow{3}*{{Baboon}}
			&PGD&21.87&11.33& &0.75&0.32&&0.13&\textbf{0.37}\\
			&BDR&19.78&10.45& &0.66&0.25&&0.17&0.38\\
			&BDR1&\textbf{22.23}&\textbf{12.28}& &\textbf{0.76}&\textbf{0.33}&&\textbf{0.12}&\textbf{0.37}\\
			\\
			\multirow{3}*{{Barba}}
			&PGD&17.91&14.03& &0.64&\textbf{0.34}&&0.20&0.38\\
			&BDR&21.35&12.89& &0.56&0.30&&0.17&0.45\\
			&BDR1&\textbf{23.61}&\textbf{14.48}& &\textbf{0.68}&\textbf{0.34}&&\textbf{0.14}&\textbf{0.37}\\
			\\
			\multirow{3}*{{Lena}}		 	 
			&PGD&23.40&13.72& &0.56&0.19&&0.14&0.40\\
			&BDR&21.12&12.58& &0.44&0.14&&0.18&0.46\\
			&BDR1&\textbf{23.60}&\textbf{14.35}& &\textbf{0.58}&\textbf{0.20}&&\textbf{0.13}&\textbf{0.38}\\
			\\
			\multirow{3}*{{Harbour}}
			
			&PGD&\textbf{22.45}&12.56& &0.62&\textbf{0.26}&&0.14&0.40\\
			&BDR&20.31&11.61& &0.52&0.20&&0.18&0.41\\
			&BDR1&22.25&\textbf{12.45}& &\textbf{0.63}&\textbf{0.26}&&\textbf{0.13}&\textbf{0.39}\\
			\\
			\multirow{3}*{{Golden Hill}}
			
			&PGD&23.83&\textbf{12.56}& &0.67&0.24&&0.13&0.41\\
			&BDR&21.50&11.61& &0.54&0.18&&0.17&0.46\\
			&BDR1&\textbf{23.85}&11.45& &\textbf{0.68}&\textbf{0.25}&&\textbf{0.12}&\textbf{0.37}\\
			\hline
			\label{T1}
	\end{tabular}}
\end{table}
\subsection{Empirical test of the background information model}
Finally, a physical setup is built to demonstrate the capacity of the background information model for the PR problem.
\begin{figure}
	\subfigure[ Measured $480\times480$ diffraction pattern.\label{spectrum}]{
		\includegraphics[width=0.3\textwidth]{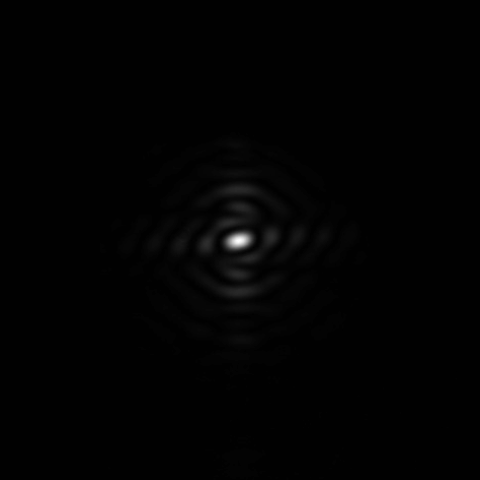}}
	\subfigure[Coherent diffraction imaging experimental setup.\label{setup}]{
		\includegraphics[width=0.3\textwidth]{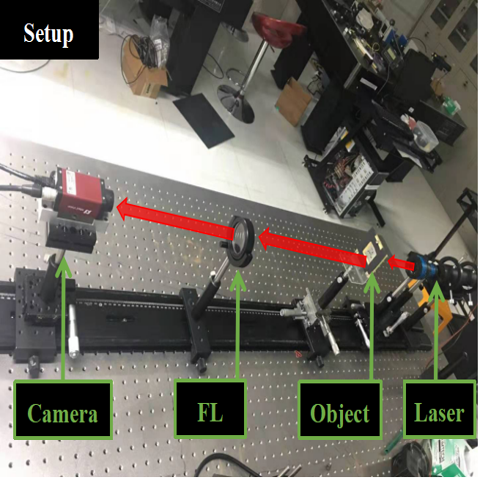}}
	\subfigure[Image recovered by $480\times480$ spectrum after denoising.\label{clearfive}]{
		\includegraphics[width=0.3\textwidth]{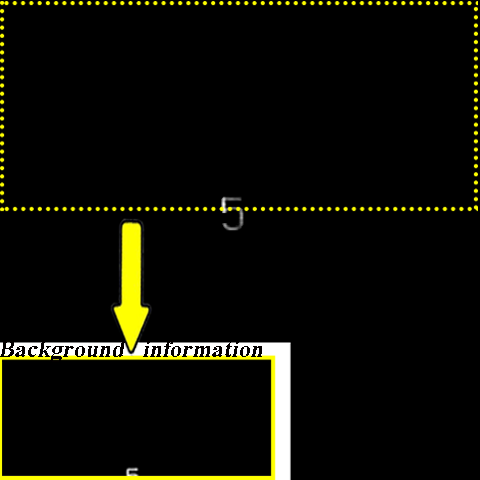}}
	\caption{Experiment setup: (a) shows the diffraction pattern measured by the CCD. (b) depicts the CDI setup. (c) displays the denoised image recovered by HIO. The yellow rectangle indicates the background information.}
\end{figure} 
 
As shown in Figure \ref{setup}, a collimated and expanded laser beam with a wavelength of $\lambda = 632.8nm$ illuminates an object that is the symbol '5' from group 1 of the negative USAF 1951 target. The size of the object is 1.2 mm\footnote{We use the character '5', which is more complex than the vertical and horizontal 3 bars in the USAF 1951, because the purpose of this experiment is to test the BDR method's ability to recover the object from its diffraction pattern, not to test the resolution of the recovered object.}. The object is manufactured by Edmund Optics. More details about the resolution targets can be found at \url{https://www.edmundoptics.cn/p/2-x-2-negative-1951-usaf-hi-resolution-target/12479/#}. The laser beam is focused by an FT lens (with a focal length of $L=300$mm and diameter of $D=30$mm). The diffraction pattern is captured by a CCD sensor (Point Grey Grasshopper2, with a pixel size of $6.45\mu m$ and $2048\times2048$ pixels) using a microscope objective (Newport, $\times10$, NA=0.25). The microscope objective is placed at the rear focal plane of the FT lens. The exposure time is $5000\mu s$.

The captured diffraction pattern has a size of $480\times480$, as shown in Figure \ref{spectrum}. Using these raw data, we have reconstructed a clear image with the HIO algorithm, as shown in Figure \ref{clearfive}. However, this reconstruction was based on a large bandwidth of the spectrum. It is crucial to devise a method to recover images with less frequency information.
For the remaining tests, we assume that the obtained region of the spectrum is $200\times200$, which can be seen in Figure \ref{smallspectrum}. As shown in Figure \ref{HIO}, using the HIO algorithm on the $200\times200$ spectrum cannot reconstruct an identifiable object image due to the loss of high-frequency information.
\begin{figure}
	\centering
	\subfigure[$200\times200$ diffraction pattern which is cropped from Figure \ref{spectrum}\label{smallspectrum}]{\includegraphics[width=0.3\textwidth]{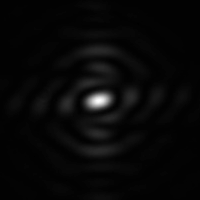}}
	\\
	\subfigure[The result of HIO from $200\times200$ spectrum after denoising.\label{HIO}]{
		\includegraphics[width=0.45\textwidth]{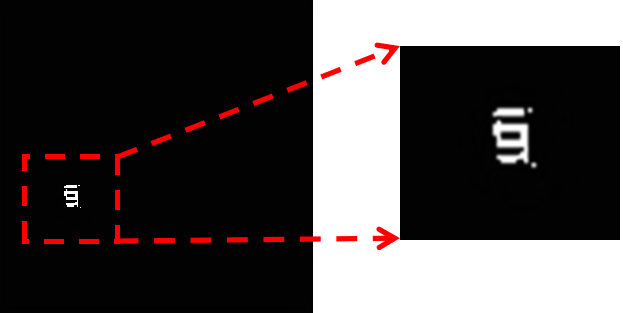}}
	\subfigure[The result of BDR from $200\times200$ spectrum after denoising.\label{BackDR}]{
		\includegraphics[width=0.45\textwidth]{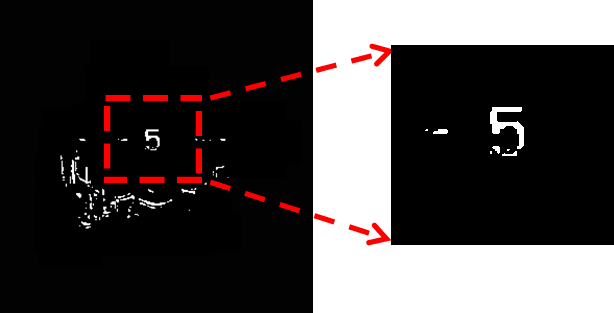}}
	\caption{Experiment results: (a) shows the diffraction pattern used to recover the image. (b) shows the result recovered via HIO. (c) shows the result recovered by the BDR method with the aid of the background information in Figure \ref{clearfive}.}
	\label{correspond}
\end{figure}  
Using the same spectrum shown in Figure \ref{smallspectrum}, we apply the background information model to recover the image. The background information, which is the part of the $200\times200$ image visible in Figure \ref{clearfive}, has a length almost equal to $0.5n$. The BDR method is then applied to recover the image, and the results are displayed in Figure \ref{BackDR}.
We observe that the outline of the number '5' can be reconstructed from the $200\times200$ spectrum with the help of the background information. Although the quality of the image recovered by the BDR algorithm is not very high, this may be attributed to insufficient background information. Nonetheless, it effectively demonstrates the advantage of the background information model.

\section{Conclusion}
This paper focuses on studying the phase retrieval with background information model. Theoretical results demonstrate that the solution's uniqueness can be guaranteed with less background information for 2-D images as compared to previous works. Additionally, the uniqueness result for the general signal is also derived, and a bound between the estimation and the ground truth is established when the measurement is corrupted by noise and suffers from the background information bias.

Two methods, BDR and CBDR, are proposed to find the ground truth. Theoretical results show that the BDR method can achieve an R-linear convergence rate if the initialization is close to the ground truth. Meanwhile, CBDR method can guarantee global convergence based on the satisfaction of the F-RIP property. Numerical simulations demonstrate the high efficiency of the BDR method, which requires less background information but offers better performance. Although the CBDR method have a better convergence guarantee, it requires more background information to avoid introducing extra solutions. The advantages of the background information model are further illustrated through the application of coherent diffraction imaging.

\section*{Acknowledgment}
This work was supported by the National Natural Science Foundation of China (61977065), National Key Research and Development Program (2020YFA0713504). The authors would like to thank Prof. Xiujian Li and Dr. Wusheng Tang for assisting in the optical tests.
\bibliography{2.bib}
\appendix
\section{The proof of Theorem \ref{mt}}\label{appendixa}
To prove the theoretical results, we introduce lemmas related to the set regularity.

\begin{definition}
	The proximal normal cone $N_{\bm{\Omega}}^{\mathrm{prox}}(\mathbf{z})$ to a set $\bm{\Omega}$ at a point $\mathbf{z}\in\bm{\Omega}$ is defined by
	$$N_{\bm{\Omega}}^{\mathrm{prox}}(\mathbf{z}):=\mathrm{cone}(\mathbb{P}^{-1}_{\bm{\Omega}}(\mathbf{z})-\mathbf{z}),$$
	\end{definition}
where $\mathbb{P}_{\bm{\Omega}}^{-1}(\mathbf{z})$ is the preimage set of $\mathbb{P}_{\bm{\Omega}}(\cdot)$ under $\mathbf{z}$.
Then the limiting normal cone or the Mordukhovich normal cone $N_{\bm{\Omega}}(\mathbf{z})$ is defined as any vector that can be the limit of the proximal normal, namely
$$
	N_{\bm{ \Omega}}(\mathbf{z}):=\{\mathbf{u}\in\mathbb{R}^{n+k}|~\exists\{\mathbf{z}^p\}\subset\bm{\Omega},\mathbf{u}^p\in N_{\bm{\Omega}}^{\mathrm{prox}}(\mathbf{ z}^p)~\text{so that}~ \mathbf{z}^p\rightarrow\mathbf{z},\mathbf{u}^p\rightarrow\mathbf{u},\text{as}~p\rightarrow\infty\}.
$$
Figure \ref{cone} provides a graphical illustration of the Mordukhovich normal cone. From its definition, we know that the limiting normal cone is the smallest cone that satisfies the following two properties:
\begin{enumerate}
\item[i.] $\mathbb{P}_{\bm{\Omega}}^{-1}(\mathbf{z}) \subseteq\left(\mathbb{I}+N_{\bm{\Omega}}\right)(\mathbf{z}).$
\item[ii.] For any sequence $\mathbf{z}^{p} \rightarrow \mathbf{z}$, where $\{\mathbf{z}^p\}\in\bm{\Omega}$, if $\mathbf{u}^{p}\in N_{\bm{\Omega}}\left(\mathbf{z}^{p}\right)$ converges to $\mathbf{u}$ as $p\rightarrow\infty$, then $\mathbf{u}$ must lie in $N_{\bm{\Omega}}(\mathbf{z}).$
\end{enumerate}
Next, we will introduce some definitions about the regularity of a set. These definitions can help us analyze the convergence rate.
\begin{definition}
	[Regularity of set systems] Let $\mathbf{A}$ and $\mathbf{B}$ be two subsets of $\mathbb{R}^{n+k}$ and let $\mathbf{z}\in\mathbf{A}\cap\mathbf{B}$. We say that the system $\{\mathbf{A},\mathbf{B}\}$ is strongly regular at $\mathbf{z}
	$ if 
	\begin{eqnarray}
	N_{\mathbf{A}}(\mathbf
	{z})\cap(-N_{\mathbf{B}}(\mathbf{z}))=\{0\}.
	\end{eqnarray}
\end{definition}
If $\mathbf{A}$ and $\mathbf{B}$ are sets defined in Theorem \ref{mt}, we can prove that $N_{\mathbf{A}}(\mathbf{z})=\{\beta\mathbf{z}\}$, where $\beta$ is some real constant, and $N_{\mathbf{B}}(\mathbf{z})=\{\mathbf{z}\in\mathbb{R}^{n+k}:\mathbf{x}=\bm{0},\mathbf{y}\in\mathbb{R}^k\}$. Then we can conclude that $	N_{\mathbf{A}}(\mathbf
{z})\cap(-N_{\mathbf{B}}(\mathbf{z}))=\{0\}$.
\begin{proof}
	First, we will prove that the set $\mathbf{ B }$ satisfies the strongly regularity condition. Recall that $\mathbf{B}=\{\mathbf{z}\in\mathbb{R}^{n+k}:z_{n+i}=y_i,i=1,...,k\}.$ Since $\mathbf{ B }$ is an affine space, every limiting sequence $\{\mathbf{z}^p\}\in\mathbf{B}\rightarrow\mathbf{ z}$ is equivalent to $\{\mathbf{x}^p\}\in\mathbb{R}^n\rightarrow\mathbf{ x}.$ Therefore, for each $\mathbf{z}^p$, $N_{\mathbf{B}}^{\mathrm{prox}}(\mathbf{z}^p)=\mathrm{cone}(\mathbb{P}_{\mathbf{ B}}^{-1}(\mathbf{z}^p)-\mathbf{z}^p)=\{\tilde{\mathbf{z}}\in\mathbb{R}^{n+k}:\mathbf{x}=\bm{0},\tilde{\mathbf{y}}\in\mathbb{ R }^k\}.$ Thus, $N_{\mathbf{B}}(\mathbf{z})=\tilde{\mathbf{ z}}$ where $\mathbf{x}=\bm{0}$, and $\tilde{\mathbf{y}}\in\mathbb{R}^k$.
	
	For $\mathbf{ A}=\{\mathbf{z}\in\mathbb{R}^{n+k}:\vert\mathbf{F}^{\mathrm{H}}\mathbf{z}\vert=\mathbf{b}^{\frac{1}{2}}\}$, because $\mathbb{P}_{\mathbf{A}}^{-1}(\mathbf{z})=\{\mathbf{z}_1\in\mathbb{R}^{n+k}:\theta(\mathbf{F}^{\mathrm{H}}\mathbf{z}_1)=\theta(\mathbf{F}^{\mathrm{H}}\mathbf{z})\},$ where $\theta(\cdot)$ is the element-wise phase operator. The structure of set $\mathbb{P}_\mathbf{A}^{-1}(\mathbf{z})$ can be derived using theoretical results from \cite{Hayes1980Phase}.
	\begin{lemma}\label{Hay}\cite{Hayes1980Phase}
Let $\mathbf{z}_1\in\mathbb{R}^n$ be a vector with a $z$-transform that has no zeros in reciprocal pairs. For any $\mathbf{z}_2\in\mathbb{R}^n$, if $m \geq 2n-1$ and $\theta(\mathbf{F}^{\mathrm{H}}\mathbf{z}_1)=\theta(\mathbf{F}^{\mathrm{H}}\mathbf{z}_2)$, then $\mathbf{z}_1=\beta\mathbf{z}_2$ for some positive constant $\beta$.
		\end{lemma}  
		Recalling the Assumption \ref{assumption}, we can conclude that $\mathbb{P}_{\mathbf{A}}^{-1}(\mathbf{z})=\{\beta\mathbf{z},\beta>0\}$ by using Lemma \ref{Hay}. As a result, $N_{\mathbf{A}}(\mathbf{z})=\{\alpha\mathbf{z}\}$, $\alpha$ is some real number. Because $\mathbf{x}\neq\bm{0}$, $\alpha\mathbf{z}$ cannot be lied in the $N_{\mathbf{B}}(\mathbf{z})$ only if $\alpha=0$. As a result,  $N_{\mathbf{A}}(\mathbf{ z})\cap(-N_{\mathbf{B}}(\mathbf{z}))=\{\mathbf{0}\}$.
\end{proof}

%
\begin{figure}
	\subfigure[$\mathbf{A}:=\{z\in\mathbb{C}:\vert z\vert=1\}$.]{
		\includegraphics[width=0.45\textwidth]{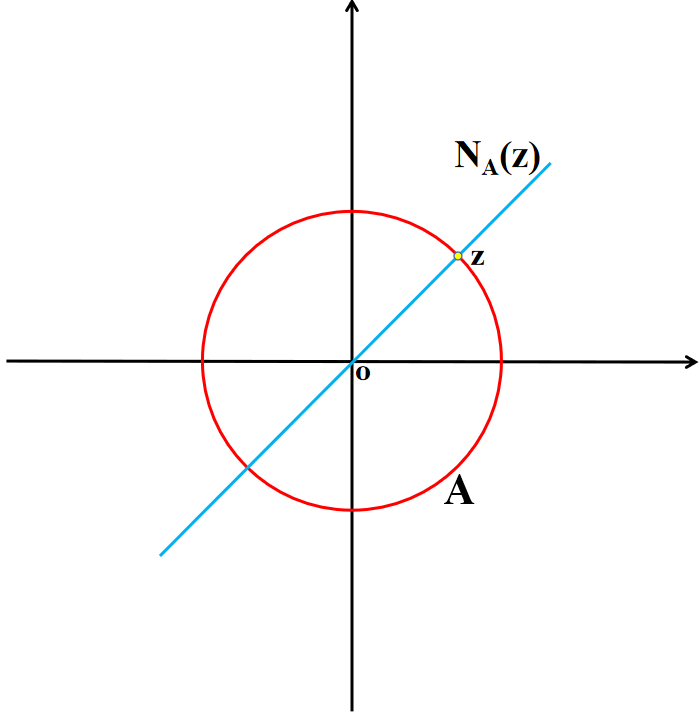}}
	\subfigure[$\mathbf{B}:=\{\mathbf{z}\in\mathbb{R}^2:z_2=1\}$.]{
		\includegraphics[width=0.45\textwidth]{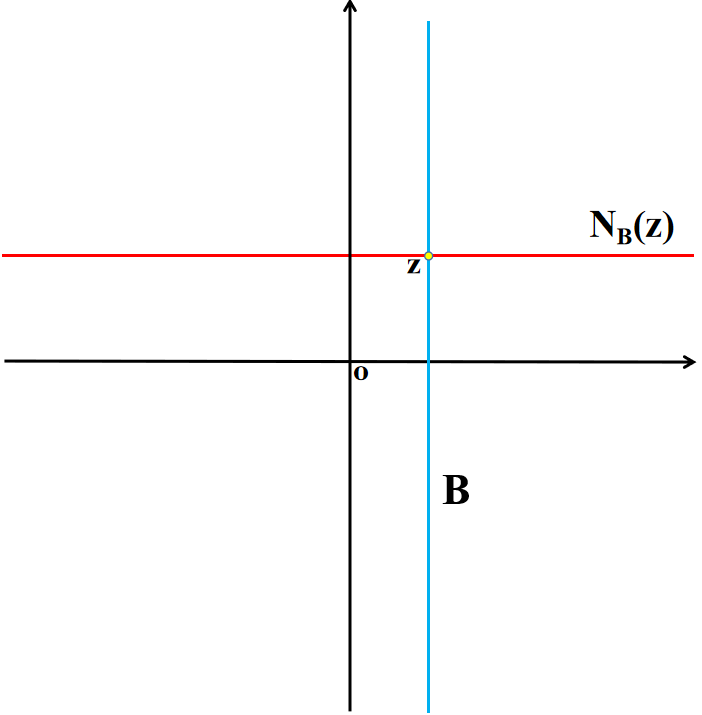}}
\caption{The graphical illustration of the Mordukhovich normal cone.}
	\label{cone}
\end{figure}

\indent
Next, the super-regularity condition will be introduced to aid in the analysis of the convergence rate of the BDR method.
\begin{definition}
	[super-regularity]Let $\bm{\Omega}$ be a closed subset of $\mathbf{R}^{n+k}$. We say $\bm{\Omega}$ is $(\varepsilon,\delta)$-regular at $\mathbf{z}$ if there exist $\varepsilon\geq0$ and $\delta\geq0$ such that
	$$\left\{\begin{array}{l}{\mathbf{z}_1, \mathbf{z}_2 \in\bm{\Omega} \cap \mathbb{B}_{\delta}(\mathbf{z})} \\ {\mathbf{u} \in N_{\bm{\Omega}}^{\operatorname{prox}}(\mathbf{z}_1)}\end{array} \right.\text{leads to}~\langle \mathbf{u}, \mathbf{z}_2-\mathbf{z}_1\rangle \leq \varepsilon\|\mathbf{u}\|_2 \cdot\|\mathbf{z}_2-\mathbf{z}_1\|_2,$$
  then we say $\bm{\Omega}$ is $(\varepsilon,\delta)$-regular at $\mathbf{z}
	$. 
	
	Moreover, if for any $\varepsilon>0$,  there exists $\delta>0$ so that $\bm{\Omega}$ is $(\varepsilon,\delta)$-regular at $\mathbf{z}$, then $\bm{\Omega}$ is said to be super-regular at $\mathbf{z}$. 
\end{definition}
Super-regularity is a property that can be thought of as a smoothness condition. In this paper, we will prove that $\mathbf{A}$ and $\mathbf{B}$ in Theorem \ref{mt} satisfy the super-regularity condition.
\begin{proof}
	For any two elements $\mathbf{z}_1$ and $\mathbf{z}_2$ in $\mathbf{B}$, as proven above, $\mathbf{u}\in N_{\mathbf{B}}^{\text{prox}}(\mathbf{z}_1)=\{\tilde{\mathbf{z}}|\mathbf{x}=\bm{0},\tilde{\mathbf{y}}\in\mathbb{R}^k\}$. As a result, $\langle\mathbf{u},\mathbf{z}_2-\mathbf{z}_1\rangle=0\leq\varepsilon\|\mathbf{u}\|_2\cdot\|\mathbf{z}_2-\mathbf{z}_1\|_2$. Thus, for any $\varepsilon>0$, there always exists $\delta>0$, which satisfies the inequality.
	
	Next, when $\mathbf{z}$ satisfies the conditions in Assumption \ref{assumption}, there are at most $2^{n+k}$ different solutions in $\mathbf{R}^{n+k}$. So $|\mathbf{A}|\leq2^{n+k}$. Therefore, there exists $\delta>0$ such that $\mathbf{ A}\cap\mathbb{B}_{\delta}(\mathbf{ z})=\mathbf{z}$. As a result, we can conclude that $\mathbf{A}$ is super-regular at $\mathbf{z}$. The proof is complete.
\end{proof}

By combining the super-regularity of feasible sets $\mathbf{A}$ and $\mathbf{B}$ with the strong regularity of $\{\mathbf{A},\mathbf{B}\}$, the local R-linear convergence rate of BDR can be deduced using Theorem 4.3 in \cite{phan2016linear}.

\textit{Remark:
Theorem \ref{mt} guarantees that the sequence $\{\mathbf{z}^p\}$ converges to $\mathbf{A}\cap\mathbf{B}$. As stated in Theorem \ref{t12}, finding the fixed points of $\mathbb{T}(\cdot)$ is sufficient to find $\mathbf{A}\cap\mathbf{B}$. Thus, we can loosen the constraints in Theorem \ref{mt} accordingly. This will be considered in future work.
It should be noted that calculating $\mathbb{P}_{\mathbf{A}}(\mathbf{z}^{p-1})$ is difficult under these conditions and usually requires iterative algorithms such as gradient descent to approximate the projection.} 

\section{The proof of the \texorpdfstring{\eqref{chaos}}{}}
We will use the techniques presented in \cite{krahmer2014suprema} to prove \eqref{chaos}. Before the proof, an important theorem will be introduced first.   

\begin{theorem}\cite{krahmer2014suprema}\label{citeKra}
	Let $\mathcal{A}$ be a set of matrices, and let $\boldsymbol{\xi}$ be a random vector whose entries $\xi_j$ are independent, mean-zero, variance 1 , and $L$-subgaussian random variables. Set
	$$
	\begin{aligned}
	& E=\gamma_2\left(\mathcal{A},\|\cdot\|_{2}\right)\left(\gamma_2\left(\mathcal{A},\|\cdot\|_{2}\right)+d_F(\mathcal{A})\right)+d_F(\mathcal{A}) d_{2}(\mathcal{A}), \\
	& V=d_{2}(\mathcal{A})\left(\gamma_2\left(\mathcal{A},\|\cdot\|_{2}\right)+d_F(\mathcal{A})\right), \quad \text { and } \quad U=d_{2}^2(\mathcal{A})
	\end{aligned}
	$$
	Then, for $t>0$,
	$$
	\mathbb{P}\left(\sup _{\boldsymbol{A} \in \mathcal{A}}\left|\|\boldsymbol{A} \boldsymbol{\xi}\|_2^2-\mathbb{E}\|\boldsymbol{A} \boldsymbol{\xi}\|_2^2\right| \geq c_1 E+t\right) \leq 2 \exp \left(-c_2 \min \left\{\frac{t^2}{V^2}, \frac{t}{U}\right\}\right) .
	$$
	The constants $c_1$ and $c_2$ depend only on $L$.
\end{theorem}

 In Theorem \ref{citeKra}, $d_F(\mathcal{A}):=\sup _{\boldsymbol{A} \in \mathcal{A}}\|\boldsymbol{A}\|_\mathrm{F}$, $d_{2}(\mathcal{A}):=\sup _{\boldsymbol{A} \in \mathcal{A}}\|\boldsymbol{A}\|_{2}$, and $\|\boldsymbol{A}\|_{2}:=\sup _{\|\boldsymbol{x}\|_2 \leq 1}\|\boldsymbol{A} \boldsymbol{x}\|_2.$ $\gamma_(\mathcal{A},\|\cdot\|_2)$ is the $\gamma_2$-functional of a set of matrices $\mathcal{A}$ endowed with the operator norm, which is defined as follows.
 
 \begin{definition}\cite{krahmer2014suprema}
     For a metric space $(T, d)$, an admissible sequence of $T$ is a collection of subsets of $T$ denoted by $\left\{T_r: r \geq 0\right\}$, such that for every $s \geq 1,\left|T_r\right| \leq 2^{2^r}$ and $\left|T_0\right|=1$. Then define
$$
\gamma_2(T, d):=\inf \sup _{t \in T} \sum_{r=0}^{\infty} 2^{r/2} d\left(t, T_r\right),
$$
where the infimum is taken with respect to all admissible sequences of $T$.
 \end{definition}
 
To prove \eqref{chaos}, we will apply Theorem \ref{citeKra} to analyze $\sup _{\boldsymbol{A} \in \mathcal{A}}\left|\|\boldsymbol{A} \boldsymbol{\xi}\|_2^2-\mathbb{E}\|\boldsymbol{A} \boldsymbol{\xi}\|_2^2\right|$ when $\mathcal{A}=\{\mathbf{V}_{\mathbf{h}},\mathbf{h}\in\mathbf{D}\}$, where $\mathbf{D}:=\left\{\boldsymbol{h} \in \mathbb{R}^{n+k}:\|\boldsymbol{h}\|_2 \leq 1,\boldsymbol{h}\in\mathbf{C}_2\right\}$, where $\mathbf{C}_2:=\{\mathbf{h}||\mathbf{F}^{\text{H}}_i\mathbf{h}|^2\leq4,i=1,\cdots,n+k\}.$
\begin{proof}
Recall that
$$\mathbf{V}_{\mathbf{h}}:=\frac{1}{\sqrt{k}}\mathbf{P}_{\bm{\Omega}}\mathbf{F}^{-1}\hat{\mathbf{H}}\mathbf{F}_2,$$
where $\mathbf{F}_2$ represents the last $k$ columns of Fourier matrix $\mathbf{F}$. 
Using the convolutional theorem, we can obtain $$\mathbf{F}^{-1}\hat{\mathbf{H}}\mathbf{F}=\mathbf{F}^{-1}\mathbf{F}(\mathbf{H}_{\text{circ}}\mathbf{I})=\mathbf{H}_{\text{circ}},$$
where each row of $\mathbf{H}_{\text{circ}}$ is a circular shift of $\mathbf{h}$ and $\mathbf{H}_{\text{circ}}$ is symmetry. Therefore, 
$$\|\mathbf{V}_{\mathbf{h}}\|_{\mathrm{F}}^2=\frac{1}{k}\|\mathbf{P}_{\bm{\Omega}}\mathbf{F}^{-1}\hat{\mathbf{H}}\mathbf{F}_2\|^2_{\mathrm{F}}\leq\frac{1}{k}\|\mathbf{P}_{\bm{\Omega}}\mathbf{F}^{-1}\hat{\mathbf{H}}\mathbf{F}\|_{\mathrm{F}}^2\leq\|\mathbf{h}\|_2^2\leq1.$$
Thus, for all $\mathbf{h}\in\mathbf{D}$, we have
$$
d_F(\mathcal{A})\leq1.
$$
At the same time, for every $\boldsymbol{h} \in \mathbf{D}$, we have
\begin{eqnarray*}
\left\|\mathbf{V}_{\mathbf{h}}\right\|_{2} &=&\frac{1}{\sqrt{k}}\left\|\mathbf{P}_{\bm{\Omega}}\mathbf{F}^{-1}\hat{\mathbf{H}}\mathbf{F}_2\right\|_{2}\\
&\leq& \frac{1}{\sqrt{k}}\|\widehat{\mathbf{H}}\|_{2}=\frac{1}{\sqrt{k}}\|\mathbf{h}\|_{\widehat{\infty}}\leq\frac{2}{\sqrt{k}},
\end{eqnarray*}
where $\|\mathbf{h}\|_{\widehat{\infty}}:=\|\mathbf{F}^{\text{H}}\mathbf{h}\|_{\infty}$. So for every $\mathbf{h}\in \mathbf{D}$, we have
$$
d_{2}(\mathcal{A}) \leq \frac{2}{\sqrt{k}}.
$$

 The $\gamma_2$-function can be bounded in terms of covering numbers by the well-known Dudley integral as below\cite{krahmer2014suprema}
\begin{eqnarray}\label{integral}    \gamma_2\left(\mathcal{A},\|\cdot\|_{2}\right) \lesssim \int_0^{d_{2}(\mathcal{A})} \log ^{1 / 2} N\left(\mathcal{A},\|\cdot\|_{2}, u\right) d u,
\end{eqnarray}
where $A \lesssim B$ means that there is an absolute constant $c_1$ such that $A \leq c_1 B$. For a metric space $(T, d)$ and $u>0$, the covering number $N(T, d, u)$ is defined as the minimum number of open balls with radius $u$ in $(T, d)$ required to cover $T$.

Because
$$
\left\|\mathbf{V}_{\mathbf{h}_1}-\mathbf{V}_{\mathbf{h}_2}\right\|_{2}=\left\|\mathbf{V}_{\mathbf{h}_1-\mathbf{h}_2}\right\|_{2} \leq\frac{1}{\sqrt{k}}\|\mathbf{h}_1-\mathbf{h}_2\|_{\widehat{\infty}}.
$$
and hence for every $u>0, N\left(\mathcal{A},\|\cdot\|_{2}, u\right) \leq N\left(\mathbf{D}, \frac{1}{\sqrt{k}}\|\cdot\|_{\widehat{\infty}}, u\right)$.

To bound the integral \eqref{integral} above, we will divide the interval of integral into two parts.
\begin{eqnarray}\label{mainintegral}
    \gamma_2\left(\mathcal{A},\|\cdot\|_{2}\right)& \lesssim& \int_0^{\frac{2}{\sqrt{k}}} \log ^{1 / 2} N\left(\mathbf{D}, k^{-1 / 2}\|\cdot\|_{\widehat{\infty}}, u\right) d u\\
    &\lesssim&\underbrace{\int_{\frac{1}{k}}^{\frac{2}{\sqrt{k}}} \log ^{1 / 2}N\left(\mathbf{D}, k^{-1 / 2}\|\cdot\|_{\widehat{\infty}}, u\right)d u}_{(1)}\nonumber\\
    &&+ \underbrace{\int_0^{\frac{1}{k}} \log ^{1 / 2} N\left(\mathbf{D}, k^{-1 / 2}\|\cdot\|_{\widehat{\infty}}, u\right) d u}_{(2)}\nonumber.
\end{eqnarray}

For (1), we use Lemma \ref{V1} to bound $\log ^{1 / 2}N\left(\mathbf{D}, k^{-1 / 2}\|\cdot\|_{\widehat{\infty}}, u\right)d u$ in the interval $[\frac{1}{k},\frac{2}{\sqrt{k}}].$
\begin{lemma}\label{V1}\cite{krahmer2014suprema}
    There exists an absolute constant c for which the following holds. Let $X$ be a normed space, consider a finite set $\mathcal{U} \subset X$ of cardinality $N$, and assume that for every $L \in \mathbb{N}$ and $\left(\boldsymbol{u}_1, \ldots, \boldsymbol{u}_L\right) \in \mathcal{U}^L, \mathbb{E}_\epsilon\left\|\sum_{j=1}^L \epsilon_j \boldsymbol{u}_j\right\|_X \leq A \sqrt{L}$, where $\left(\epsilon_j\right)_{j=1}^L$ denotes a Rademacher vector. Then for every $u>0$, $$
\log N\left(\operatorname{conv}(\mathcal{U}),\|\cdot\|_X, u\right) \leq c(A / u)^2 \log N.
$$
\label{lemma3.4}
\end{lemma}
The proof of Lemma \ref{lemma3.4} can be shown in \cite{krahmer2014suprema}. It is applied for the set
$$
\mathcal{U}=\left\{ \pm \sqrt{2} e_1, \ldots, \pm \sqrt{2} e_{n+k}\right\}.
$$
where the $e_i$ are the standard basis vectors. Noting that $\mathbf{D}\subset\mathbf{B}_2^{n+1}\subset \operatorname{conv}(\mathcal{U})$ and that we may choose $A=c \sqrt{\log (n+k)}$, so we have

\begin{eqnarray}\label{boundpart1}
\log N\left(\mathbf{D}, k^{-1 / 2}\|\cdot\|_{\widehat{\infty}}, u\right) & \leq \log N\left(\operatorname{conv}(\mathcal{U}), k^{-1 / 2}\|\cdot\|_{\widehat{\infty}}, u\right) \nonumber\\
& \lesssim\left(\frac{1}{u\sqrt{k}}\right)^2 \log ^2(n+k) .
\end{eqnarray}

Using the bound \eqref{boundpart1} to calculate part (2) will result in an improper integral. Therefore, we need to provide a new bound of $\log N\left(\mathbf{D}, k^{-1 / 2}\|\cdot\|_{\widehat{\infty}}, u\right)$. To do so, we apply the standard volumetric argument and obtain the following lemma:
\begin{lemma}
    Let $\|\cdot\|$ be some semi-norm on $\mathbb{R}^n$ and let $U$ be a subset of the unit ball $\mathbf{B}_1^n$. Then the covering numbers satisfy, for $t>0$,
    $$
N(U,\|\cdot\|, t) \leq\left(1+\frac{2}{t}\right)^n.
$$
\end{lemma}

Because $\mathbf{D}\subset\mathbf{B}_2^{n+k}\subset\mathbf{B}_1^{n+k}$, and $k^{-1 / 2}\|\cdot\|_{\widehat{\infty}}$ is a semi-norm, we have
\begin{eqnarray}\label{boundpart2}
\log N\left(\mathbf{D}, k^{-1 / 2}\|\cdot\|_{\widehat{\infty}}, u\right) \lesssim (n+k)\log (1+\frac{2}{u}).
\end{eqnarray}
Combing \eqref{boundpart1} and \eqref{boundpart2} into \eqref{mainintegral}, we have
\begin{eqnarray}
    \gamma_2\left(\mathcal{A},\|\cdot\|_{2}\right)& \lesssim& \int_0^{d_{2}(\mathcal{A})} \log ^{1 / 2} N\left(\mathbf{D}, k^{-1 / 2}\|\cdot\|_{\widehat{\infty}}, u\right) d u\nonumber\\
    &\lesssim&\int_{\frac{1}{k}}^{\frac{2}{\sqrt{k}}} \log ^{1 / 2}N\left(\mathbf{D}, k^{-1 / 2}\|\cdot\|_{\widehat{\infty}}, u\right)d u\nonumber \\
    &&+\int_0^{\frac{1}{k}} \log ^{1 / 2} N\left(\mathbf{D}, k^{-1 / 2}\|\cdot\|_{\widehat{\infty}}, u\right) d u\nonumber\\
    &\lesssim&\int_{\frac{1}{k}}^{\frac{2}{\sqrt{k}}} \left(\frac{1}{u\sqrt{k}}\right) \log (n+k) d u \nonumber\\
    &&+\sqrt{(n+k)}\int_0^{\frac{1}{k}}\sqrt{\log (1+\frac{2}{u})} d u\nonumber\\
&\lesssim&\frac{\log(n+k)\log(2\sqrt{k})}{\sqrt{k}}+\frac{\sqrt{n+k}}{\sqrt{2}k}\sqrt{\log(e(1+2k))},\label{sq1}
\end{eqnarray}
where the inequality of $\int_{0}^{\frac{1}{k}}\sqrt{\log (1+\frac{2}{u})}du$ above is held by the Cauthy-Schwarz.

$$\int_{0}^{\frac{1}{k}}\sqrt{\log (1+\frac{2}{u})}du\leq \sqrt{\int_{0}^{\frac{1}{k}}1du\int_{0}^{\frac{1}{k}}\log (1+\frac{2}{u})du}=\frac{1}{\sqrt{k}} \sqrt{\int_{0}^{\frac{1}{k}}\log (1+\frac{2}{u})du}.$$

A change of variables and integration by parts yields
\begin{eqnarray*}
\int_0^\frac{1}{k} \log \left(1+\frac{2}{u}\right) d u&=&\frac{1}{2}\int_{k}^{\infty} t^{-2} \log (1+2t) d t \\
& =&-\frac{1}{2}\left.t^{-1} \log (1+2t)\right|_{k} ^{\infty}+\int_{k}^{\infty}\frac{1}{2t^2+t} d t\\ 
&\leq& \frac{1}{2k} \log \left(1+2k\right)+\int_{k}^{\infty} \frac{1}{2t^2} d t \\
& =&\frac{1}{2k} \log \left(e(1+2k)\right).
\end{eqnarray*}

\begin{figure}
		\includegraphics[width=0.7\textwidth]{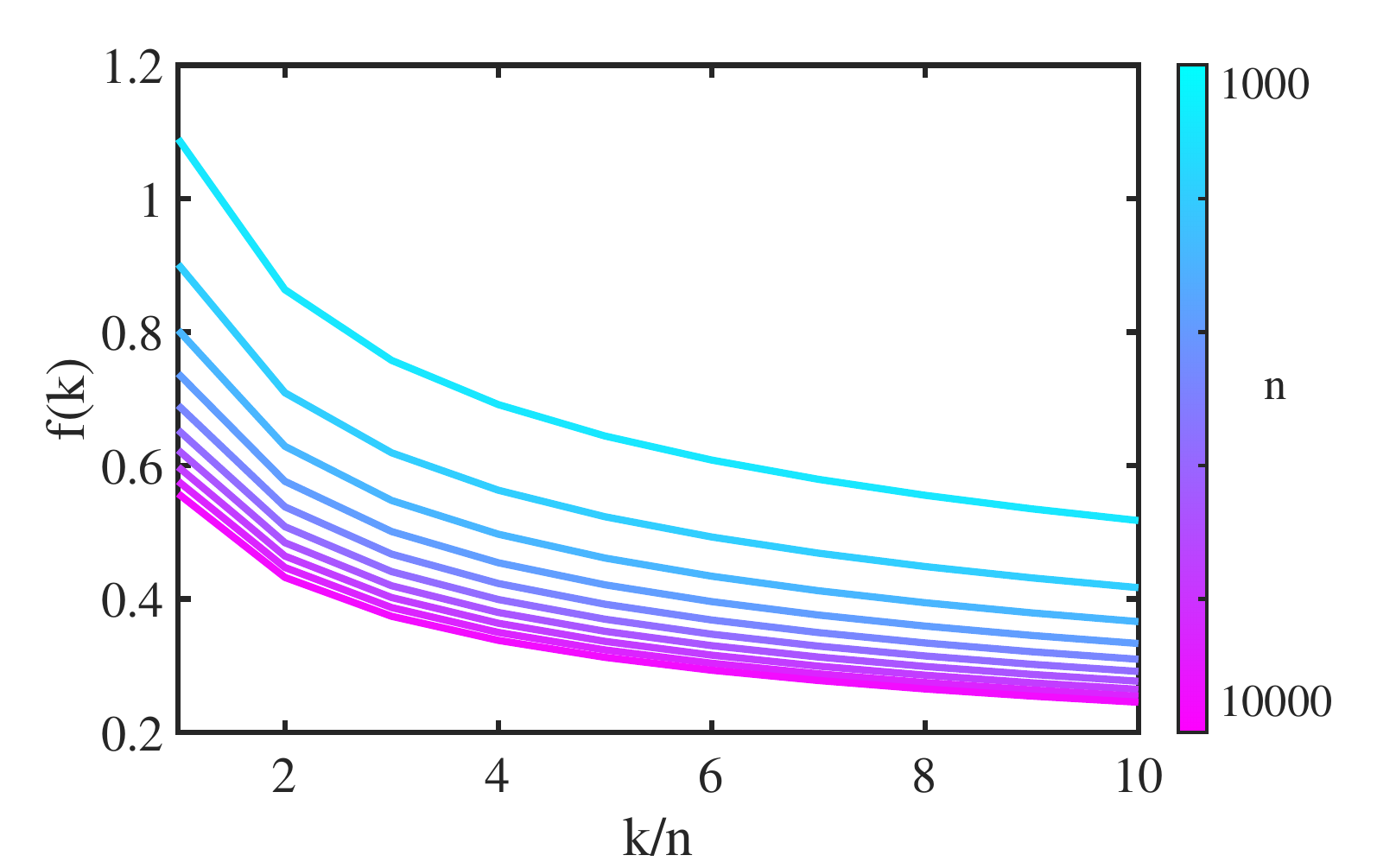}
\centering
 \caption{The plot shows $f(k)$ as $k/n$ increases from $1$ to $10$ with an interval of $1$, and $n$ varies from $1000$ to $10000$.}	\label{plot}
\end{figure} 

Using L'Hospital's rule, we can derive that $\lim_{k\rightarrow+\infty}\gamma_2(\mathcal{A},\|\cdot\|_2)=0$. Therefore, for any $\delta>0$, as long as $k$ is sufficiently large, we can guarantee that $\gamma_2\left(\mathcal{A},\|\cdot\|_{2}\right) \lesssim \frac{\sqrt{\delta}}{2}$ for the given value of $k$. Define $$f(k):=\frac{\log(n+k)\log(2\sqrt{k})}{\sqrt{k}}+\frac{\sqrt{n+k}}{\sqrt{2}k}\sqrt{\log(e(1+2k))}.$$
We plot the $f(k)$ in Fig.\ref{plot}, where $k/n$ increases from $1$ to $10$ with an interval of $1$, and $n$ varies from $1000$ to $10000$. We can find that $f(k)<1$ as long as $k/n$ is large enough.

Using the definition of $E$, we can derive that $\lim_{k\rightarrow+\infty}E=0$. Therefore, by properly choosing the constant $c$ in Theorem \ref{mainth} such that $k$ is sufficiently large, we have:
$$
E \leq \frac{\delta}{4 c_1},
$$
where $E$ and $c_1$ are chosen as in Theorem \ref{citeKra}. 
By Theorem \ref{citeKra} and let $t=\delta/4$, we obtain

$$
\mathbb{P}\left(\delta_s \geq \frac{\delta}{2}\right) \leq \mathbb{P}\left(\delta_s \geq c_1 E+\delta / 4\right) \leq \exp \left(-c_2\left(\frac{k\delta^2}{16(\sqrt{\delta}+2)^2}\right)\right) \leq \eta,
$$
which completes the proof after possibly increasing the value of $c$ sufficiently to compensate for $c_2$.
\end{proof}
\section{The proof of the \texorpdfstring{\eqref{cheby}}{}}

Next, we will evaluate that 

$$\sup _{\boldsymbol{h} \in\mathbf{C}_2}\left|\frac{1}{k}\left(\mathbf{P}_{\bm{\Omega}}\mathbf{F}^{-1}\hat{\mathbf{H}}\mathbf{F}_1\mathbf{x}\right)^{\text{H}}\mathbf{P}_{\bm{\Omega}}\mathbf{F}^{-1}\hat{\mathbf{H}}\mathbf{F}_2\mathbf{y}\right|.$$
By the Chebyshev inequality, we can conclude that
\begin{eqnarray*}
\mathbb{P}(\left|\frac{1}{k}\left(\mathbf{P}_{\bm{\Omega}}\mathbf{F}^{-1}\hat{\mathbf{H}}\mathbf{F}_1\mathbf{x}\right)^{\text{H}}\mathbf{P}_{\bm{\Omega}}\mathbf{F}^{-1}\hat{\mathbf{H}}\mathbf{F}_2\mathbf{y}\right|\geq\frac{\delta}{2})&\leq&\frac{4\left\|\left(\mathbf{P}_{\bm{\Omega}}\mathbf{F}^{-1}\hat{\mathbf{H}}\mathbf{F}_2\right)^{\text{H}}\mathbf{P}_{\bm{\Omega}}\mathbf{F}^{-1}\hat{\mathbf{H}}\mathbf{F}_1\mathbf{x}\right\|^2_2}{\sigma^2k^2}\\
&\leq&\frac{64\left\|\mathbf{x}\right\|^2_2}{\delta^2k^2}\leq\eta.    
\end{eqnarray*}

\end{document}